\documentclass[11pt]{article}
\usepackage{amssymb,amsfonts,amsmath,amsthm,url}
\usepackage{paralist,xspace}
\usepackage{simplemargins}

\usepackage{epsfig,graphicx,wrapfig}

\newtheorem{theorem}{Theorem}[section]
\newtheorem{corollary}[theorem]{Corollary}

\newtheorem{proposition}[theorem]{Proposition}
\newtheorem{lemma}[theorem]{Lemma}

\newenvironment{definition}[1][Definition.]{\begin{trivlist}
\item[\hskip \labelsep {\bfseries #1}]}{\end{trivlist}}

\def\P{\mathcal{P}}
\def\A{\mathcal{A}}
\def\B{\mathcal{B}}
\def\F{\mathbb{F}}
\def\C{\mathcal{C}}
\def\od{\stackrel{\mathrm{def}}{=}}
\def\U{\mathcal{U}}
\def\RR{\mathbb{R}}
\def\supp{\operatorname{supp}}
\def\p{{\bf p}}
\def\q{{\bf q}}

\usepackage{multicol}
\usepackage{color}
\definecolor{gold}{rgb}{0.85,.66,0}
\definecolor{cherry}{rgb}{0.9,.1,.2}
\definecolor{burgundy}{rgb}{0.8,.2,.2}
\definecolor{orangered}{rgb}{0.85,.3,0}
\definecolor{orange}{rgb}{0.85,.4,0}
\definecolor{olive}{rgb}{.45,.4,0}
\definecolor{lime}{rgb}{.6,.9,0}
\definecolor{green}{rgb}{.2,.7,0}
\definecolor{grey}{rgb}{.4,.4,.2}
\definecolor{brown}{rgb}{.4,.2,.1}
\definecolor{blue}{rgb}{0,.0, .81}
\definecolor{bluepurple}{rgb}{.3, .0, .7}

\begin{document}

\title{The neural ring: an algebraic tool for analyzing\\ the intrinsic structure of neural codes}
\author{Carina Curto, Vladimir Itskov, Alan Veliz-Cuba, Nora Youngs\\
{\small Department of Mathematics, University of Nebraska-Lincoln}}
\date{\small May 21, 2013}
\maketitle

\vspace{.5in}
\begin{abstract}
\noindent Neurons in the brain represent external stimuli via neural codes.   These codes often arise from stereotyped stimulus-response maps, associating
to each neuron a convex receptive field.  An important problem confronted by the brain is to infer properties of a represented stimulus space
without knowledge of the receptive fields, using only the intrinsic structure of the neural code.  How does the brain do this?  To address this question,
it is important to determine what stimulus space features can -- in principle -- be extracted from neural codes.  This motivates us to define the {\it neural ring}
and a related {\it neural ideal}, algebraic objects that encode the full combinatorial data of a neural code.  Our main finding is that these objects can be expressed in a
``canonical form'' that directly translates to a minimal description of the receptive field structure intrinsic to the code.  
We also find connections to Stanley-Reisner rings, and use ideas similar to those in the theory of monomial ideals to obtain an
algorithm for computing the primary decomposition of pseudo-monomial ideals.
This allows us to algorithmically extract the canonical form associated to any neural code, providing the 
groundwork for inferring stimulus space features from neural activity alone. 
\end{abstract}

\vspace{.6in}

\tableofcontents

\vspace{.6in}

\section{Introduction}
Building accurate representations of the world is one of the basic functions of the brain.  It is well-known that when a stimulus is paired with pleasure
or pain, an animal quickly learns the association. Animals also learn, however, the (neutral) relationships between stimuli of the same type.
For example, a bar held at a 45-degree angle appears more similar to one held at 50 degrees than to a perfectly vertical one.  
Upon hearing a triple of distinct pure tones, one seems to fall ``in between'' the other two.  An explored environment is perceived not as a collection of
disjoint physical locations, but as a spatial map.  In summary, we do not experience the world as a stream of unrelated stimuli; rather, our brains
organize different types of stimuli into highly structured {\em stimulus spaces}.  

The relationship between neural activity and stimulus space structure has, nonetheless, received remarkably little attention.  
In the field of neural coding, 
much has been learned about the coding properties of individual neurons
by investigating stimulus-response functions, such as place fields \cite{OKeefeDostrovsky,PathIntegration}, orientation tuning curves \cite{WatkinsBerkley74, Ben-Yishai1995}, and other examples of ``receptive fields'' obtained by measuring neural activity in response to experimentally-controlled stimuli.  Moreover, numerous studies have shown that neural 
activity, together with knowledge of the appropriate stimulus-response functions, can be used to accurately estimate a newly presented
stimulus \cite{Brown98,Deneve99,Ma}.  This paradigm is being actively extended and revised to include information
present in populations of neurons, spurring debates on the role of correlations in neural coding 
\cite{NirenbergLatham03,Averbeck2006,SchneidmanBialek2006}.
In each case, however, the underlying structure of the stimulus space is
assumed to be known, and is not treated as itself emerging from the activity of neurons.  This approach is particularly problematic when
one considers that the brain does {\em not} have access to stimulus-response functions, and
must represent the world {\em without} the aid of dictionaries that lend meaning to neural activity \cite{gap}.  In coding theory parlance, the brain 
does not have access to the encoding map, and must therefore represent stimulus spaces via the intrinsic structure of the neural code.

How does the brain do this?  In order to eventually answer this question, we must first tackle a simpler one:  
\medskip

\noindent{\bf Question:} \textit{What can be 
inferred about the underlying stimulus space from neural activity alone?} I.e., what stimulus space features are encoded in the {\em intrinsic structure} of the neural code, and can thus be extracted without knowing the individual stimulus-response functions?
\medskip

\noindent Recently we have shown that, in the case of hippocampal place cell codes, certain topological features of 
the animal's environment can be inferred from the neural code alone, without knowing the place fields \cite{gap}.  
As will be explained in the next section, this information can be extracted from a simplicial
complex associated to the neural code.  {\it What other stimulus space features can be inferred from the neural code?}
For this, we turn to algebraic geometry.
Algebraic geometry provides a useful framework for inferring geometric and topological characteristics
of spaces by associating rings of functions to these spaces.  All relevant features of the underlying space are encoded in the intrinsic structure of the ring, where coordinate functions become indeterminates, and the space itself is defined in terms of ideals in the ring.  Inferring features of a space from properties of functions -- without specified domains -- is similar to the task confronted by the brain, so it is natural to expect that this framework may shed light on our question.

In this article we introduce the \textit{neural ring}, an algebro-geometric object that can be associated to any combinatorial neural code.
Much like the simplicial complex of a code, the neural ring encodes information about the underlying stimulus space in a way that discards specific knowledge of receptive field maps, and thus gets closer to the essence of how the brain might represent stimulus spaces.  Unlike the simplicial complex, the neural ring retains the full combinatorial data of a neural code, packaging this data in a more computationally tractable manner.  We find that this object, together with a closely related
{\em neural ideal}, can be used to algorithmically extract a compact, minimal description of the {\it receptive field structure} dictated by the code.  This enables us to more directly
tie combinatorial properties of neural codes to features of the underlying stimulus space, a critical step towards answering our motivating question.

Although the use of an algebraic construction such as the neural ring is quite novel in the context of neuroscience, the neural code (as we define it) is at its core a combinatorial
object, and there is a rich tradition of associating algebraic objects to combinatorial ones \cite{MillerSturmfels}.  The most well-known example is perhaps the Stanley-Reisner ring \cite{StanleyBook}, which turns out to be closely related to the neural ring.  Within mathematical biology, associating polynomial ideals to combinatorial data has also been fruitful.  Recent examples include inferring wiring diagrams in gene-regulatory networks \cite{Laubenbacher2007, Alan2012} and applications to chemical reaction networks \cite{ShiuSturmfels2010}.   Our work also has parallels to the study of design ideals in algebraic statistics \cite{alg-stats-book}. 

The organization of this paper is as follows.  In Section~\ref{sec:neural-codes} we introduce {\it receptive field codes}, and explore how the requirement of convexity enables these codes to constrain the structure of the underlying stimulus space.  In Section~\ref{sec:neural-ring} we define the neural ring and the neural ideal, and find explicit relations that enable us to compute these objects for any neural code.  Section~\ref{sec:RF-structure} is the heart of this paper.   Here we present an alternative set of relations for the neural ring, and demonstrate how they enable us to ``read off'' receptive field structure from the neural ideal.  We then introduce pseudo-monomials and pseudo-monomial ideals, by analogy to monomial ideals; this allows us to define a natural ``canonical form'' for the neural ideal.  Using this, we can extract minimal relationships among receptive fields that are dictated by the structure of the neural code.
Finally, we present an algorithm
for finding the canonical form of a neural ideal, and illustrate how to use our formalism for inferring receptive field structure in a detailed example.  Section~\ref{sec:prim-decomp} describes the primary decomposition of the neural ideal and, more generally, of pseudo-monomial ideals.  Computing the primary decomposition of the neural ideal is a critical step in our canonical form algorithm, and it also yields a natural decomposition of the neural code in terms of intervals of the Boolean lattice.  We end this section with an algorithm for finding the primary decomposition of any pseudo-monomial ideal, using ideas similar to those in the theory of square-free monomial ideals.  
All longer proofs can be found in Appendix 1.   A detailed classification of neural codes on three neurons is given in Appendix 2.

\section{Background \& Motivation}\label{sec:neural-codes}

\subsection{Preliminaries}
In this section we introduce the basic objects of study: neural codes, receptive field codes, and convex receptive field codes.  We then discuss various ways in which
the structure of a convex receptive field code can constrain the underlying stimulus space.  
These constraints emerge most obviously from the simplicial complex of a neural code, but (as will be made clear) there are also constraints that arise from aspects
of a neural code's structure that go well beyond what is captured by the simplicial complex of the code.  

Given a set of neurons labelled $\{1, \dots, n\} \od [n]$, we define a {\em neural
  code} $\C \subset \{0,1\}^n$ as a set of binary patterns of neural activity.  
An element of a neural code is called a {\em codeword}, $c = (c_1,\ldots,c_n) \in \C$, and corresponds to a subset of
neurons
$$\supp(c) \od \{i \in [n] \mid c_i = 1\} \subset [n].$$
Similarly, the entire code $\C$ can be identified with a set of subsets of neurons,
$$\supp \C \od \{\supp(c) \mid c \in \C\} \subset 2^{[n]},$$
where $2^{[n]}$ denotes the set of all subsets of $[n]$.
Because we discard the
details of the precise timing and/or rate of neural activity, what we
mean by {neural code} is often referred to in the neural coding
literature as a {\em combinatorial code} \cite{BialekBerry,Bialek2008}.

A set of subsets $\Delta \subset 2^{[n]}$
is an (abstract) {\em simplicial complex} if $\sigma \in \Delta$ and $\tau \subset \sigma$ implies $\tau \in \Delta$.  We will say that a neural code $\C$ is a simplicial complex if $\supp \C$ is a simplicial complex.  In cases where the code is {\it not} a simplicial complex, we can complete the code
to a simplicial complex by simply adding in missing subsets of codewords.  This allows us to define the {\em simplicial complex of the code} as
$$\Delta(\C) \od \{\sigma \subset [n] \mid \sigma \subseteq \supp(c) \text{ for some } c \in \C\}.$$
Clearly, $\Delta(\C)$ is the smallest simplicial complex that contains $\supp \C$.

\subsection{Receptive field codes (RF codes)}
Neurons in many brain areas have activity patterns that can be characterized by receptive fields.\footnote{In the vision literature, 
the term ``receptive field'' is reserved for subsets of the visual field; we use the term in a more general sense, applicable to any modality.}
Abstractly, a {\it receptive field} is a map $f_i:X \rightarrow \RR_{\geq 0}$ from a space of stimuli, $X$, to the average firing rate of a single neuron, $i$, in response to each stimulus.   Receptive fields are computed by correlating neural
responses to independently measured external stimuli.  We follow a common abuse of language, where both the map and its support (i.e., the subset $U_i \subset X$ where $f_i$ takes on
positive values) are referred to as ``receptive fields.''  {\it Convex} receptive fields are convex\footnote{A subset $B\subset \RR^n$ is {\it convex} if, given any pair of points $x,y\in B$, the point $z=tx+(1-t)y$ is contained in $B$ for any $t\in [0,1].$}
subsets of the stimulus space, for $X \subset \RR^d$. 
The paradigmatic examples are orientation-selective neurons in visual cortex \cite{WatkinsBerkley74, Ben-Yishai1995} and 
hippocampal place cells \cite{OKeefeDostrovsky,PathIntegration}.  Orientation-selective neurons have {\it tuning curves} that reflect a neuron's preference for a particular angle (see Figure 1A). Place cells are neurons that have {\it place fields}; i.e., each neuron has a preferred (convex) region of the animal's physical environment where it has a high firing rate (see Figure 1B).  Both tuning curves and place fields are examples of receptive fields.

\begin{figure}[h]\label{fig1}
\begin{center}
\includegraphics[width=5in]{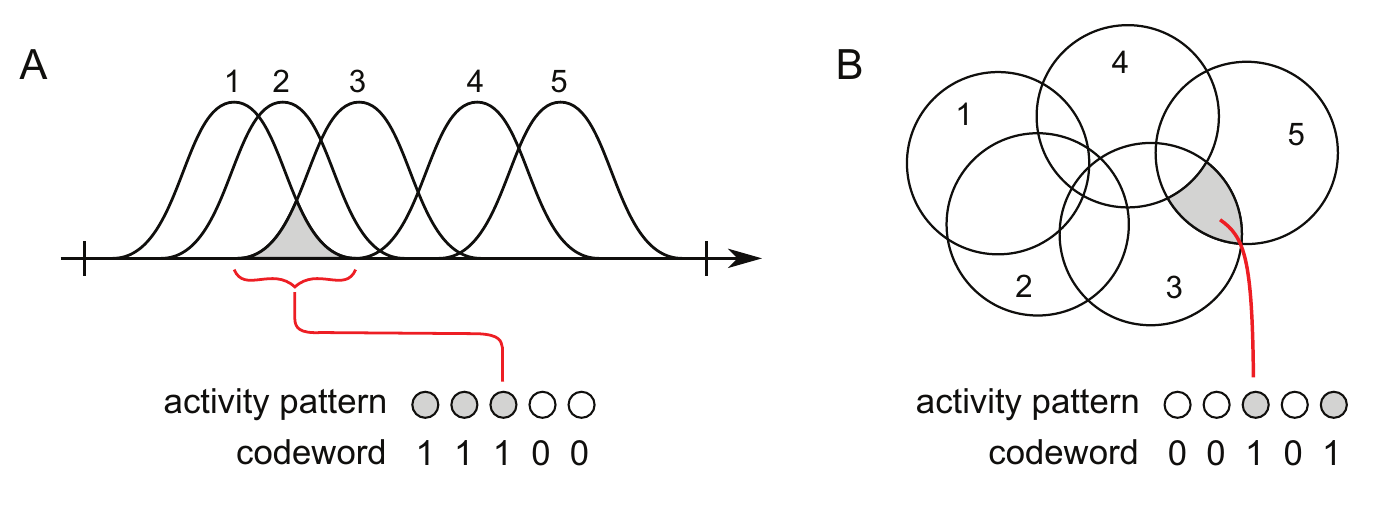}
\end{center}
\vspace{-.25in}
\caption{\small {Receptive field overlaps determine codewords in 1D and 2D RF codes.}  (A) Neurons in a 1D RF code have receptive fields that overlap on a line segment (or circle, in the case of orientation-tuning).  Each stimulus on the line
corresponds to a binary codeword.  Gaussians depict graded firing rates for neural responses; this additional information is discarded by the RF code.  (B) Neurons in a 2D RF code, such as a place field code, have receptive fields that partition a two-dimensional stimulus space into non-overlapping regions, as illustrated by the shaded area.  All stimuli within one of these regions will activate the same set of neurons, and hence have the same corresponding codeword. }
\end{figure}

A receptive field code (RF code) is a neural code that corresponds to the brain's representation of the 
stimulus space covered by the receptive fields.  When a stimulus lies in the intersection of several receptive
fields, the corresponding neurons may co-fire while the rest
remain silent.  The active subset $\sigma$ of neurons can be identified with a binary
codeword $c \in \{0,1\}^n$ via $\sigma = \supp(c)$.  Unless otherwise noted, a stimulus space $X$ need only be a topological space.
However, we usually have in mind $X \subset \RR^d$, and this becomes important when we consider {\it convex} RF codes.

\begin{definition}
Let $X$  be a stimulus space (e.g., $X \subset \RR^d$), and let $\U = \{U_1,\ldots,U_n\}$ be a collection of open sets, 
with each $U_i \subset X$ the receptive field of the $i$-th neuron in a population of $n$ neurons.  
The {\em receptive field code (RF code)} $\C(\U) \subset \{0,1\}^n$ is the set of all binary codewords corresponding to stimuli in $X$: 
$$\C(\U) \od \{ c \in \{0,1\}^n \mid (\bigcap_{i \in \supp(c)} U_i) \setminus (\bigcup_{j \notin \supp(c)} U_j) \neq \emptyset \}.$$
If $X \subset \RR^d$ and each of the $U_i$s is also a {\it convex} subset of $X$, then we say that $\C(\U)$ is a {\em convex} RF code.
\end{definition}

Our convention is that the empty intersection is $\bigcap_{i \in \emptyset} U_i = X$, and the empty union is $\bigcup_{i \in \emptyset} U_i = \emptyset$.
This means that if $\bigcup_{i=1}^n U_i \subsetneq X$, then $\C(\U)$ includes the all-zeros codeword corresponding to an ``outside'' point not covered by the receptive fields; 
on the other hand, if
$\bigcap_{i=1}^n U_i \neq \emptyset$, then $\C(\U)$ includes the all-ones codeword. 
Figure 1 shows examples of convex receptive fields covering one- and two-dimensional stimulus spaces, 
and examples of codewords corresponding to regions defined by the receptive fields.

Returning to our discussion in the Introduction, we have the following question:
If we can assume $\C = \C(\U)$ is a RF code, then {\it what can be learned about the underlying stimulus space $X$ from knowledge only of $\C$, and not of $\U$?} 
The answer to this question will depend critically on whether or not we can assume that the RF code is convex.  In particular, if we don't assume convexity of the receptive fields, then any code can be realized as a RF code in any dimension.

\begin{lemma}\label{lemma:RFform}
Let $\C \subset \{0,1\}^n$ be a neural code.  Then, for any $d \geq 1$, there exists a stimulus space $X \subset \RR^d$ and a collection of open sets $\U = \{U_1,\ldots,U_n\}$ (not necessarily convex), with $U_i \subset X$ for each $i \in [n]$, such that $\C = \C(\U)$.
\end{lemma}

\begin{proof}
Let $\C\subset\{0,1\}^n$ be any neural code, and order the elements of $\C$ as $\{c_1,\ldots,c_m\}$, where $m = |\C|$. For each $c\in \C$, choose a distinct point $x_c\in \RR^d$ and an open neighborhood $N_c$ of $x_c$ such that no two neighborhoods intersect.
Define $U_j \od \bigcup_{j\in \supp(c_k)} N_{c_k}$, let $\U=\{U_1,\ldots,U_n\}$, and $X=\bigcup_{i=1}^m N_{c_i}$.  Observe that if the all-zeros codeword is in $\C$, then
$N_{\bf 0} = X \setminus \bigcup_{i=1}^{n} U_i$ corresponds to the ``outside point'' not covered by any of the $U_i$s.  By construction, $\C = \C(\U).$
\end{proof}

Although any neural code $\C \subseteq \{0,1\}^n$ can be realized as a RF code, it is {\it not} true that any code can be realized as a {\it convex} RF code.  Counterexamples
can be found in codes having as few as three neurons.

\begin{lemma}\label{lemma:convex-counterexample}
The neural code $\C = \{0,1\}^3 \setminus \{111, 001\}$ on three neurons cannot be realized as a convex RF code.
\end{lemma}

\begin{wrapfigure}{r}{.33\linewidth}
\vspace{-.3in}
 \includegraphics[width=2in]{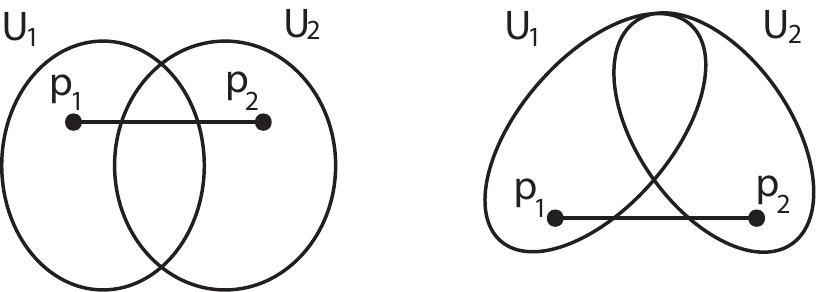} 
 \caption{\small Two cases in the proof of Lemma~\ref{lemma:convex-counterexample}.}
 \vspace{0in}
\end{wrapfigure}

\text{             }

\vspace{-.3in}
\begin{proof}

Assume the converse, and let $\U = \{U_1,U_2, U_3\}$ be a set of convex open sets in  $\RR^d$ such that $\C = \C(\U)$.  The code
necessitates that $U_1 \cap U_2 \neq \emptyset$ (since $110 \in \C$), $(U_1 \cap U_3) \setminus U_2 \neq \emptyset$ (since $101 \in \C$), and 
$(U_2 \cap U_3) \setminus U_1 \neq \emptyset$ (since $011 \in \C$).  Let $p_1 \in (U_1 \cap U_3) \setminus U_2$ and $p_2 \in (U_2 \cap U_3) \setminus U_1.$
Since $p_1,p_2 \in U_3$ and $U_3$ is convex, the line segment $\ell = (1-t)p_1 + t p_2$ for $t \in [0,1]$ must also be contained in $U_3$.  There are just two possibilities.  Case 1: $\ell$ passes through $U_1 \cap U_2$ (see Figure 2, left).  This implies $U_1 \cap U_2 \cap U_3 \neq \emptyset$, and hence $111 \in \C$, a contradiction.  Case 2: $\ell$ does not intersect $U_1 \cap U_2$.  Since $U_1, U_2$ are open sets, this implies $\ell$ passes outside of $U_1 \cup U_2$ (see Figure 2, right),
and hence $001 \in \C$, a contradiction.  
\end{proof}

\subsection{Stimulus space constraints arising from convex RF codes}

It is clear from Lemma~\ref{lemma:RFform} that there is essentially no constraint on the stimulus space for realizing a code as a RF code.  However, if we demand that $\C$
is a {\it convex} RF code, then the overlap structure of the $U_i$s sharply constrains the geometric and topological properties of the underlying stimulus space $X$.   
To see how this works, we first consider the simplicial complex of a neural code, $\Delta(\C)$.
Classical results in convex geometry and topology provide constraints on the underlying stimulus space $X$ for convex RF codes, based on the structure of $\Delta(\C)$.  
We will discuss these next.  We then turn to the question of constraints that arise from combinatorial properties of a neural code $\C$ that are {\it not} captured
by $\Delta(\C)$.

\subsubsection{Helly's theorem and the Nerve theorem}\label{sec:helly-nerve}

Here we briefly review two classical  and well-known theorems in convex geometry and topology, Helly's theorem and the Nerve theorem, as they apply to convex RF codes.
Both theorems can be used to relate the structure of the simplicial complex of a code, $\Delta(\C)$, to topological features of the underlying stimulus space $X$.  

Suppose $\U = \{U_1,\ldots,U_n\}$ is a finite collection of convex open subsets of $\RR^d$, with dimension $d<n$.  
We can associate to $\U$ a simplicial complex $N(\U)$ called the {\it nerve} of $\U$.  A subset $\{i_1,..,i_k\}\subset [n]$ belongs to $N(\U)$ if and only if  the appropriate intersection $\bigcap_{\ell=1}^kU_{i_\ell} $  is  nonempty. 
If we think of the $U_i$s as receptive fields, then $N(\U) = \Delta(\C(\U))$.  In other words, the nerve of the cover corresponds to the simplicial complex of the associated (convex) RF code.
\medskip

\noindent {\bf Helly's theorem.}
{\em Consider $k$ convex subsets, $U_1,\ldots,U_k \subset \RR^d,$ for $d<k$.
If the intersection of every $d+1$ of these sets is nonempty, then the full intersection $\bigcap_{i=1}^k U_i$ is also nonempty.}
\medskip

\noindent A nice exposition of this theorem and its consequences can be found in \cite{helly-review}.
One straightforward consequence is that the nerve $N(\U)$ is completely determined by its $d$-skeleton, and corresponds to the largest simplicial complex with that $d$-skeleton.  For example, if $d = 1$, then $N(\U)$ is a clique complex (fully determined by its underlying graph).  Since $N(\U) = \Delta(\C(\U))$, Helly's theorem imposes constraints on the minimal dimension of the stimulus space $X$ when $\C = \C(\U)$ is assumed to be a convex RF code.

\medskip

\noindent {\bf Nerve theorem.} 
{\em The homotopy type of $X(\U) \od \bigcup_{i=1}^n U_i$ is equal to the homotopy type of the nerve of the cover, $N(\U)$.  In particular, $X(\U)$ and $N(\U)$ have exactly the same homology groups.}
\medskip

\noindent The Nerve theorem  is an easy consequence of \cite[Corollary 4G.3]{Hatcher}.  This is a powerful theorem relating the simplicial complex of a RF code, $\Delta(\C(\U)) = N(\U)$, to topological features of the underlying space, such as homology groups and other homotopy invariants.  
%
%
Note, however, that the similarities between $X(\U)$ and $N(\U)$ only go so far.  In particular, $X(\U)$ and $N(\U)$ typically have very different {dimension}.  It is also important to keep in mind that the Nerve theorem concerns the topology of  $X(\U) = \bigcup_{i=1}^n U_i$.  In our setup, if the stimulus space $X$ is larger, so that  $\bigcup_{i=1}^n U_i \subsetneq X$, then the Nerve theorem tells us only about the homotopy type of $X(\U)$, not of $X$.  Since the $U_i$ are open sets, however, conclusions about the dimension of $X$ can still be inferred.

In addition to Helly's theorem and the Nerve theorem, there is a great deal known about $\Delta(\C(\U))=N(\U)$ for collections of convex sets in $\RR^d$.  In particular, the $f$-vectors of such simplicial complexes have been completely characterized by G. Kalai in \cite{kalai1,kalai2}.

\subsubsection{Beyond the simplicial complex of the neural code}\label{sec:beyond}

We have just seen how the simplicial complex of a neural code, $\Delta(\C)$, yields constraints on the stimulus space $X$ if we assume $\C$ can be realized as a convex RF code.  The example described in Lemma~\ref{lemma:convex-counterexample}, however, implies that other kinds of constraints on $X$ may emerge from the combinatorial structure of a neural code, even if there is no obstruction stemming from $\Delta(\C)$.

\begin{wrapfigure}{r}{.5\linewidth}
\centering
   \includegraphics[width=2.75in]{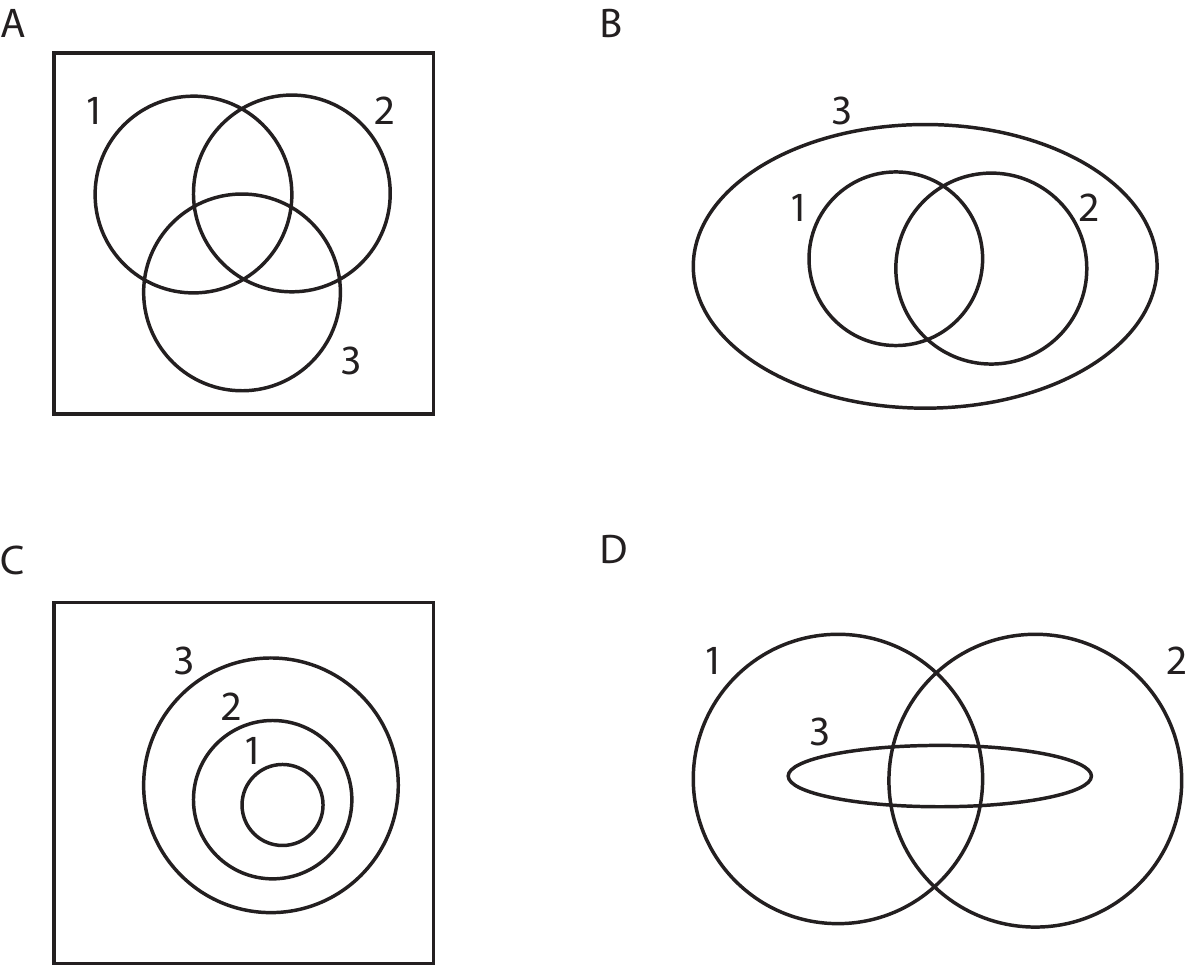} 
   \caption{\small Four arrangements of three convex receptive fields, $\U = \{U_1, U_2, U_3\}$, each having $\Delta(\C(\U)) = 2^{[3]}$.  
   Square boxes denote the stimulus space $X$ in cases where
   $U_1\cup U_2 \cup U_3 \subsetneq X$.
   (A) $\C(\U) = 2^{[3]}$, including the all-zeros codeword $000$.  (B) $\C(\U) = \{111, 101, 011, 001\}$, with $X = U_3$.
   (C) $\C(\U) = \{111, 011, 001, 000\}$.  (D) $\C(\U) = \{111, 101, 011, 110, 100, 010\},$ and $X = U_1 \cup U_2$.
The minimal embedding dimension for the codes in panels A and D is $d=2$, while for panels B and C it is $d=1$.}
\vspace{-.2in}
\end{wrapfigure}

In Figure 3 we show four possible arrangements of three convex receptive fields in the plane.  Each convex RF code has the same corresponding simplicial complex $\Delta(\C) = 2^{[3]}$, since $111 \in \C$ for each code.  Nevertheless, the arrangements clearly have different combinatorial properties.  In Figure 3C, for instance, we have $U_1 \subset U_2 \subset U_3$, while Figure 3A has no special containment relationships among the receptive fields.  This ``receptive field structure'' (RF structure) of the code has impliciations for the underlying stimulus space.

Let $d$ be the minimal integer for which the code can be realized as a convex RF code in $\RR^d$; we will refer to this as the {\em minimal embedding dimension} of $\C$.  Note that the codes in Figure 3A,D have $d=2$, whereas the codes in Figure 3B,C have $d=1$.  The simplicial complex, $\Delta(\C)$, is thus not sufficient
to determine the minimal embedding dimension of a convex RF code, but this information {\it is} somehow present in the RF structure of the code.   Similarly, in 
Lemma~\ref{lemma:convex-counterexample} we saw that $\Delta(\C)$ does not provide
sufficient information to determine whether or not $\C$ can be realized as a convex RF code; after working out
the RF structure, however, it was easy to see that the given code was not realizable.

\subsubsection{The receptive field structure (RF structure) of a neural code}

As we have just seen, the intrinsic structure of a neural code contains information about
the underlying stimulus space that cannot be inferred from the simplicial complex of the code alone. 
This information is, however, present in what we have loosely referred to as the ``RF structure'' of the code.
We now explain more carefully what we mean by this term.

Given a set of receptive fields $\U = \{U_1,\ldots,U_n\}$ in a stimulus space $X$, there are certain containment relations between intersections and unions of the $U_i$s that
are ``obvious,'' and carry no information about the particular arrangement in question.  For example, $U_1 \cap U_2 \subseteq U_2 \cup U_3 \cup U_4$ is always guaranteed to be true, because it follows from $U_2 \subseteq U_2.$  On the other hand, a relationship such as $U_3 \subseteq U_1 \cup U_2$ (as in Figure 3D) is {\it not} always present, and thus reflects something about the structure of a particular receptive field arrangement.

Let $\C \subset \{0,1\}^n$ be a neural code, and let $\U = \{U_1,\ldots,U_n\}$ be any arrangement of receptive fields in a stimulus space $X$ such that $\C = \C(\U)$ (this is guaranteed to exist by Lemma~\ref{lemma:RFform}).  The {\it RF structure} of $\C$ refers to the set of relations among the $U_i$s that are not ``obvious,'' 
and have the form:
$$\bigcap_{i \in \sigma} U_i \subseteq \bigcup_{j \in \tau} U_j, \;\; \text{ for } \;\; \sigma \cap \tau = \emptyset.$$
In particular, this includes any empty intersections $\bigcap_{i \in \sigma} U_i = \emptyset$ (here $\tau = \emptyset$).
In the Figure 3 examples, the panel A code has no RF structure relations; while panel B has $U_1 \subset U_3$ and $U_2 \subset U_3$; panel C has $U_1 \subset U_2 \subset U_3$; and panel D has $U_3 \subset U_1 \cup U_2$.

The central goal of this paper is to develop a method to algorithmically extract a minimal description of the RF structure directly from a neural code $\C$, without first
realizing it as $\C(\U)$ for some arrangement of receptive fields.  We view this as a first step towards inferring stimulus space features that cannot be obtained from the simplicial complex $\Delta(\C)$.  To do this we turn to an algebro-geometric framework, that of neural rings and ideals.  These objects are defined in Section~\ref{sec:neural-ring} so as to capture the full combinatorial data of a neural code, but in a way that allows us to naturally and algorithmically infer a compact description of the desired RF structure, as shown in Section~\ref{sec:RF-structure}.

\section{Neural rings and ideals}\label{sec:neural-ring}

In this section we define the neural ring $R_\C$ and a closely-related neural ideal, $J_\C$.  First, we briefly review some basic algebraic geometry background 
needed throughout this paper.

\subsection{Basic algebraic geometry background}

The following definitions are standard (see, for example, \cite{cox-little-oshea}). 

\begin{definition}[Rings and ideals.]
Let $R$ be a commutative ring. 
A subset $I\subseteq R$ is an {\it ideal} of $R$ if it has the following properties:
\begin{enumerate}
\item[(i)] $I$ is a subgroup of $R$ under addition.
\item[(ii)] If $a\in I$, then $ra\in I$ for all $r \in R$.
\end{enumerate}
An ideal $I$ is said to be {\it generated by} a set $A$, and we write $I=\langle A\rangle$, if $$I=\{  r_1a_1+\cdots + r_n a_n \, |\, a_i\in A, r_i\in R, \text{ and } n\in \mathbb{N}\}.$$ 
In other words, $I$ is the set of all finite combinations of elements of $A$ with coefficients in $R$.

An ideal $I \subset R$ is {\it proper} if $I\subsetneq R$.  An ideal $I \subset R$ is {\it prime} if it is proper and satisfies: if $rs\in I$ for some $r,s\in R$, then 
$r\in I$ or $s\in I$.  An ideal $m \subset R$ is {\it maximal} if it is proper and for any ideal $I$ such that $m\subseteq I\subseteq R$, either $I=m$ or $I=R$.
 An ideal $I \subset R$ is {\it radical} if $r^n\in I$ implies $r\in I$, for any $r \in R$ and $n \in  \mathbb{N}$.  An ideal $I \subset R$ is {\it primary} if $rs\in I$ implies $r\in I$ or $s^n \in I$ for some $n \in  \mathbb{N}$.  A {\it primary decomposition} of an ideal $I$ expresses $I$ as an intersection of finitely many primary ideals.
\end{definition}

\begin{definition}[Ideals and varieties.]
Let $k$ be a field, $n$ the number of neurons, and  $k[x_1,\ldots,x_n]$ a polynomial ring with one indeterminate $x_i$ for each neuron.
We will consider $k^n$ to be the neural activity space, where each point $v = (v_1,\ldots,v_n) \in k^n$ is a vector tracking the state $v_i$ of each neuron.   Note that any
polynomial $f \in k[x_1,\ldots,x_n]$ can be evaluated at a point $v \in k^n$ by setting $x_i = v_i$ each time $x_i$ appears in $f$.  We will denote this value $f(v)$.

Let $J \subset k[x_1,\ldots,x_n]$ be an ideal, and define the variety
$$V(J) \od \{v \in k^n \mid f(v) = 0 \text{ for all } f \in J\}.$$
Similarly, given a subset $S \subset k^n$, we can define the ideal of functions that vanish on this subset as
$$I(S) \od \{f \in k[x_1,\ldots.,x_n] \mid f(v) = 0 \text{ for all } v \in S\}.$$
The ideal-variety correspondence \cite{cox-little-oshea} gives us the usual order-reversing relationships: $I \subseteq J \Rightarrow V(J) \subseteq V(I)$, and
 $S \subseteq T \Rightarrow I(T) \subseteq I(S)$.  Furthermore, $V(I(V)) = V$ for any variety $V$, but
 it is not always true that $I(V(J)) = J$ for an ideal $J$ (see Section~\ref{sec:lemma-proofs}).
We will regard neurons as having only two states, ``on'' or ``off,'' and thus choose $k = \F_2 = \{0,1\}$.  
\end{definition}

\subsection{Definition of the neural ring}

Let $\C \subset \{0,1\}^n = \F_2^n$ be a neural code, and define the ideal $I_\C$ of $\F_2[x_1,\ldots,x_n]$ corresponding to the set of polynomials that vanish on all codewords in $\C$:
$$I_\C \od I(\C) = \{f \in \F_2[x_1,\ldots,x_n] \mid f(c) = 0 \text{ for all } c \in \C\}.$$
By design, $V(I_\C) = \C$ and hence $I(V(I_\C)) = I_\C$.
Note that the ideal generated by the  {\em Boolean relations},
$$\B \od \langle x_1^2-x_1,\ldots,x_n^2-x_n \rangle,$$
is automatically contained in $I_\C$, irrespective of $\C$.

The {\em neural ring} $R_\C$ corresponding to the code $\C$ is the quotient ring
$$R_\C \od \F_2[x_1,\ldots,x_n]/I_\C,$$
together with the set of indeterminates $x_1,\ldots,x_n$.  We say that two neural rings are {\em equivalent} if there
is a bijection between the sets of indeterminates that yields a ring homomorphism.
\medskip

\noindent {\bf Remark.} Due to the Boolean relations, any element $y \in R_\C$ satisfies $y^2 = y$ (cross-terms vanish because $2=0$ in $\F_2$), so the neural ring is a {\it Boolean ring} isomorphic to $\F_2^{|\C|}$.  It is important to keep in mind, however, that $R_\C$ comes equipped with a privileged set of functions, $x_1,\ldots,x_n$; this allows the ring to keep track of considerably more structure than just the size of the neural code.

\subsection{The spectrum of the neural ring}\label{sec:spec}
We can think of $R_\C$ as the ring of functions of the form $f:\C \rightarrow \F_2$ on the neural code, where each function assigns a $0$ or $1$
to each codeword $c \in \C$ by evaluating $f \in \F_2[x_1,\ldots,x_n]/I_\C$ through the substitutions $x_i = c_i$ for $i = 1,\ldots,n$.  Quotienting the original
polynomial ring by $I_\C$ ensures that there is only one zero function in $R_\C$.  The spectrum of the neural ring, $\mathrm{Spec}(R_\C)$, consists 
of all prime ideals in $R_\C$.  We will see shortly that the elements of $\mathrm{Spec}(R_\C)$ are in one-to-one correspondence with the
elements of the neural code $\C$.  Indeed, our definition of $R_\C$ was designed for this to be true.

For any point $v \in \{0,1\}^n$ of the neural activity space, let
$$m_v \od I(v) =  \{f \in \F_2[x_1,\ldots,x_n] \mid f(v) = 0 \}$$
be the maximal ideal of $\F_2[x_1,\ldots,x_n]$ consisting of all functions that vanish on $v$.  
We can also write $m_v = \langle x_1-v_1,\ldots,x_n-v_n\rangle$ (see Lemma~\ref{lemma:mv} in Section~\ref{sec:lemma-proofs}).
Using this, we can characterize the spectrum of the neural ring.

\begin{lemma}  \label{lemma:spec}
$\mathrm{Spec}(R_\C)  = \{ \bar{m}_v \mid v \in \C\},$ where $\bar{m}_v$ is the quotient of $m_v$ in $R_\C$.
\end{lemma}

\noindent The proof is given in Section~\ref{sec:lemma-proofs}.  Note that because $R_\C$ is a Boolean ring, the maximal ideal spectrum and the prime ideal spectrum coincide.

\subsection{The neural ideal \& an explicit set of relations for the neural ring}

The definition of the neural ring is rather impractical, as it does not give us explicit relations for generating $I_\C$ and $R_\C$.
Here we define another ideal, $J_\C$, via
an explicit set of generating relations.  Although $J_\C$ is closely related to $I_\C$, it
turns out that $J_\C$ is a more convenient object to study, which is why we will use the term {\it neural ideal} to refer to $J_\C$ rather than $I_\C$.

For any $v \in \{0,1\}^n$, consider the function $\rho_v \in \F_2[x_1,\ldots,x_n]$ defined as
$$\rho_v \od \prod_{i=1}^n(1-v_i-x_i) = \prod_{\{i\,|\,v_i=1\}}x_i\prod_{\{j\,|\,v_j=0\}}(1-x_j) = 
\prod_{i \in \supp(v)}x_i\prod_{j \notin \supp(v)}(1-x_j).$$
Note that $\rho_v(x)$ can be thought of as a characteristic function for $v$, since it satisfies $\rho_v(v) = 1$ and $\rho_v(x) = 0$ for any other $x \in \F_2^n$.  Now consider the ideal $J_\C \subseteq  \F_2[x_1,\ldots,x_n]$ generated by all functions $\rho_v$, for $v \notin \C$:
$$J_\C \od \langle \{ \rho_v \mid v \notin \C\}\rangle.$$
We call $J_\C$ the {\em neural ideal} corresponding to the neural code $\C$.  If $\C = 2^{[n]}$ is the complete code, we simply set $J_\C = 0$, the zero ideal.  $J_\C$ is related to $I_\C$ as follows, giving us explicit relations for the neural ring.

\begin{lemma} \label{lemma:explicit-relations}
Let $\C \subset \{0,1\}^n$ be a neural code.  Then,
$$I_\C = J_\C + \B = \big \langle \{ \rho_v \mid v \notin \C\}, \{x_i(1-x_i) \mid i \in [n]\} \big \rangle,$$
where $\B = \langle \{x_i(1-x_i) \mid i \in [n]\} \rangle$ is the ideal generated by the Boolean relations, and $J_\C$
is the neural ideal.
\end{lemma}

\noindent The proof is given in Section~\ref{sec:lemma-proofs}.

\section{How to infer RF structure using the neural ideal}\label{sec:RF-structure}

This section is the heart of the paper.  We begin by presenting an alternative set of relations that can be used to define the neural ring.  These relations
enable us to easily interpret elements of $I_\C$ as receptive field relationships, clarifying the connection between the neural ring and ideal and the RF structure of the code.
We next introduce pseudo-monomials and pseudo-monomial ideals, and use these notions to obtain a minimal description of the neural ideal, which we call
the ``canonical form.''  Theorem~\ref{thm:canonical-form} enables us to use the canonical form of $J_\C$ in order to ``read off'' a minimal description of the RF structure of the code.
Finally, we present an algorithm that inputs a neural code $\C$ and outputs the canonical form $CF(J_\C)$, and illustrate its use in a detailed example.

\subsection{An alternative set of relations for the neural ring}
Let $\C \subset \{0,1\}^n$ be a neural code, and recall by Lemma~\ref{lemma:RFform} that $\C$ can always be realized as a RF code $\C = \C(\U)$, provided we don't require the $U_i$s to be convex.
Let $X$ be a stimulus space and $\U = \{U_i\}_{i=1}^n$ a collection of open sets in $X$, and consider the RF code $\C(\U)$.
The neural ring corresponding to this code is $R_{\C(\U)}.$

Observe that the functions $f \in R_{\C(\U)}$ can be evaluated at any point $p \in X$ by assigning
$$x_i(p) = \left\{\begin{array}{cc} 1 & \text{if}\; p \in U_i \\ 0 & \text{if}\; p \notin U_i \end{array}\right.$$
each time $x_i$ appears in the polynomial $f$.  The vector $(x_1(p),\ldots,x_n(p)) \in \{0,1\}^n$ represents the neural response to the stimulus $p$.
Note that if $p \notin \bigcup_{i=1}^n U_i$, then $(x_1(p),\ldots,x_n(p)) = (0,\ldots,0)$ is the all-zeros codeword.
For any $\sigma \subset [n]$, define
$$U_\sigma \od \bigcap_{i \in \sigma} U_i, \;\text{  and  }\; x_\sigma \od \prod_{i \in \sigma} x_i.$$
Our convention is that $x_\emptyset = 1$ and $U_{\emptyset} = X$, even in cases where $X \supsetneq \bigcup_{i = 1}^n U_i$.
Note that for any $p \in X$,
$$x_\sigma(p) = \left\{\begin{array}{cc} 1 & \text{if}\; p \in U_\sigma \\ 0 & \text{if}\; p \notin U_\sigma. \end{array}\right.$$

The relations in $I_{\C(\U)}$ encode the combinatorial data of $\U$.  For example, if $U_\sigma = \emptyset$ then we cannot have $x_\sigma = 1$ at any point of the stimulus space $X$, and
must therefore impose the relation $x_\sigma$
to ``knock off'' those points.
On the other hand, if $U_\sigma \subset U_i \cup U_j,$ then $x_\sigma = 1$ implies either $x_i = 1$ or $x_j = 1$, something that is guaranteed by imposing the relation
$x_\sigma(1-x_i)(1-x_j)$.  These observations lead us to an alternative ideal, $I_\U \subset \F_2[x_1,\ldots,x_n]$, defined directly
from the arrangement of receptive fields $\U = \{U_1,\ldots,U_n\}$:
$$I_\U \od \big \langle \big\{ x_\sigma \prod_{i \in \tau} (1-x_i) \mid U_\sigma \subseteq
\bigcup_{i \in \tau}U_i \big\} \big \rangle .$$
Note that if $\tau = \emptyset$, we only get a relation for $U_\sigma = \emptyset$, and this is $x_\sigma$.
If $\sigma = \emptyset$, then $U_\sigma = X$, and we only get relations of this type if $X$ is contained in the union of the $U_i$s.  This is equivalent to the requirement that there is no ``outside point'' corresponding to the all-zeros codeword.

Perhaps unsurprisingly, it turns out that $I_\U$ and $I_{\C(\U)}$ exactly coincide, so $I_\U$ provides an alternative set of relations that can be used to define $R_{\C(\U)}$.

\begin{theorem} \label{thm:ideal-equivalence}
$I_\U = I_{\C(\U)}.$
\end{theorem}

\noindent The proof is given in Section~\ref{sec:proof1}.  

\subsection{Interpreting neural ring relations as receptive field relationships}\label{sec:types}

Theorem~\ref{thm:ideal-equivalence} suggests that we can interpret elements of $I_\C$ in terms of relationships between receptive fields.  

\begin{lemma}\label{lemma:corresp}
Let $\C \subset \{0,1\}^n$ be a neural code, and let $\U = \{U_1,\ldots,U_n\}$ be any collection of open sets (not necessarily convex) in a stimulus space $X$ 
such that $\C = \C(\U)$. 
Then, for any pair of subsets $\sigma,\tau \subset [n]$,
$$x_\sigma \prod_{i \in \tau} (1-x_i) \in I_\C \; \Leftrightarrow \; U_\sigma \subseteq \bigcup_{i \in \tau} U_i.$$
\end{lemma}

\begin{proof}  ($\Leftarrow$) This is a direct consequence of Theorem~\ref{thm:ideal-equivalence}.
($\Rightarrow$)  We distinguish two cases, based on whether or not $\sigma$ and $\tau$ intersect.
If $x_\sigma \prod_{i \in \tau} (1-x_i) \in I_\C$ and $\sigma \cap \tau \neq \emptyset$, then $x_\sigma \prod_{i \in \tau} (1-x_i) \in \B$, where 
$\B = \langle \{x_i(1-x_i) \mid i \in [n]\} \rangle$ is the ideal generated by the Boolean relations.  Consequently, the relation does not give us any information about the code,
and $U_\sigma \subseteq \bigcup_{i \in \tau} U_i$ follows trivially from the observation that  $ U_i \subseteq U_i$ for any $i \in \sigma \cap \tau$.
If, on the other hand, $x_\sigma \prod_{i \in \tau} (1-x_i) \in I_\C$ and $\sigma \cap \tau = \emptyset$, then $\rho_v \in I_\C$ for each $v \in \{0,1\}^n$ such that $\supp(v) \supseteq \sigma$
and $\supp(v) \cap \tau = \emptyset$.  Since $\rho_v(v) = 1$, it follows that $v \notin \C$ for any $v$ with  $\supp(v) \supseteq \sigma$
and $\supp(v) \cap \tau = \emptyset$.  To see this,
recall from the original definition of $I_\C$ that for all $c \in \C$, $f(c) = 0$ for any $f \in I_\C$; it follows
that $\rho_v(c) = 0$ for all $c \in \C$.  Because $\C = \C(\U)$, the fact that $v \notin \C$ for any $v$ such that $\supp(v) \supseteq \sigma$
and $\supp(v) \cap \tau = \emptyset$ implies 
$\bigcap_{i \in \sigma} U_i \setminus \bigcup_{j \in \tau} U_j = \emptyset.$
We can thus conclude that $U_\sigma \subseteq \bigcup_{j \in \tau} U_j.$
\end{proof}

\noindent Lemma~\ref{lemma:corresp} allows us to extract RF structure from the different types of relations that appear in $I_\C$:
\begin{itemize}
\item Boolean relations: $\{x_i(1-x_i)\}$.  The relation
$x_i(1-x_i)$ corresponds to $U_i \subseteq U_i$, which does not contain any information about the code $\C$.
\item Type 1 relations: $\{x_\sigma\}$.  The relation $x_\sigma$ corresponds to $U_\sigma = \emptyset$.
\item Type 2 relations: $\big\{ x_\sigma \prod_{i \in \tau} (1-x_i) \mid \sigma,\tau \neq \emptyset, \;\sigma \cap \tau = \emptyset, \; U_\sigma \neq \emptyset \text{ and } 
\bigcup_{i \in \tau} U_i \neq X \big \}$.  \\The relation $x_\sigma \prod_{i \in \tau} (1-x_i)$
corresponds to $U_\sigma \subseteq \bigcup_{i \in \tau} U_i$.
\item Type 3 relations: $\big\{\prod_{i \in \tau} (1-x_i)\big\}$. The relation $\prod_{i \in \tau} (1-x_i)$ corresponds to $X \subseteq \bigcup_{i \in \tau} U_i$.
\end{itemize}
The somewhat complicated requirements on the Type 2 relations ensure that they do not include polynomials that are multiples of Type 1, Type 3, or Boolean relations.  Note that the constant polynomial $1$ may appear as both a Type 1 and a Type 3 relation, but only if $X = \emptyset$.  The four types of relations listed above are otherwise disjoint.   Type 3 relations only appear if $X$ is fully covered by the receptive fields, and there is thus no all-zeros codeword corresponding to an ``outside'' point.

Not all elements of $I_\C$ are one of the above types, of course, but we will see that these are sufficient to generate $I_\C$.  This follows from the observation (see Lemma~\ref{lemma:gen-types}) that the neural ideal $J_\C$ is generated by the Type 1, Type 2 and Type 3 relations, and recalling that $I_\C$ is obtained from $J_\C$ be adding in the Boolean relations (Lemma~\ref{lemma:explicit-relations}).
At the same time, not all of these relations are necessary to generate the neural ideal.
Can we eliminate redundant relations to come up with
a ``minimal'' list of generators for $J_\C$, and hence $I_\C$, 
that captures the essential RF structure of the code?
This is the goal of the next section.

\subsection{Pseudo-monomials \& a canonical form for the neural ideal}\label{sec:canonical-form}

The Type 1, Type 2, and Type 3 relations are all products of linear terms of the form $x_i$ and $1-x_i$, and are thus very similar to monomials.  By analogy with square-free monomials and square-free monomial ideals \cite{MillerSturmfels}, we define the notions of pseudo-monomials and pseudo-monomial ideals.
Note that we do not allow repeated indices in our definition of pseudo-monomial, so the Boolean relations are explicitly excluded.

\begin{definition}
If $f \in \F_2[x_1,\ldots,x_n]$ has the form $f = \prod_{i \in \sigma} x_i \prod_{j\in \tau} (1-x_j)$ for some $\sigma,\tau \subset [n]$ with $\sigma\cap \tau=\emptyset$, then we say that $f$ is a {\em pseudo-monomial}.  
\end{definition}

\begin{definition}
An ideal $J \subset \F_2[x_1,\ldots,x_n]$ is a {\em pseudo-monomial ideal} if $J$ can be generated by a finite set of pseudo-monomials.  
\end{definition}

\begin{definition}
Let $J \subset \F_2[x_1,\ldots,x_n]$ be an ideal, and $f \in J$ a pseudo-monomial.  We say that $f$ is a {\em minimal} pseudo-monomial of $J$ if there does not exist another pseudo-monomial $g \in J$ with $\deg(g) < \deg(f)$ such that $f = hg$ for some $h \in \F_2[x_1,\ldots,x_n]$.
\end{definition}

\noindent By considering the set of {\it all} minimal pseudo-monomials in a pseudo-monomial ideal $J$, we obtain a unique and compact description of $J$, which we call the ``canonical form'' of $J$.

\begin{definition}
We say that a pseudo-monomial ideal $J$ is in {\em canonical form} if we present it as $J = \langle f_1,\ldots,f_l \rangle$, where the set $CF(J) \od \{f_1, \ldots, f_l\}$ is the set of {\it all} minimal pseudo-monomials of $J$.  Equivalently, we refer to $CF(J)$ as the {\em canonical form} of $J$.  
\end{definition}

\noindent Clearly, for any pseudo-monomial ideal $J \subset \F_2[x_1,\ldots,x_n]$, $CF(J)$ is unique and $J = \langle CF(J) \rangle$.  On the other hand, it is important to keep in mind that although $CF(J)$ consists of minimal pseudo-monomials, it is not necessarily a minimal set of generators for $J$.  To see why, consider the pseudo-monomial ideal $J = \langle x_1(1-x_2), x_2(1-x_3) \rangle.$ 
This ideal in fact contains a third minimal pseudo-monomial:
$x_1(1-x_3)=(1-x_3)\cdot[x_1(1-x_2)]+x_1\cdot[x_2(1-x_3)].$
It follows that $CF(J) = \{ x_1(1-x_2), x_2(1-x_3), x_1(1-x_3) \}$, but clearly we can remove $x_1(1-x_3)$ from this set and still generate $J$.

For any code $\C$, the neural ideal $J_\C$ is a pseudo-monomial ideal because $J_\C = \langle \{\rho_v \mid v \notin \C\}\rangle$, and each of the $\rho_v$s is a pseudo-monomial.  (In contrast, $I_\C$ is rarely a pseudo-monomial ideal, because it is typically necessary to include the Boolean relations as generators.)
Theorem~\ref{thm:canonical-form} describes the canonical form of $J_\C$.
In what follows, we say that $\sigma \subseteq [n]$ is {\it minimal with respect to} property $P$ if $\sigma$ satisfies $P$, but $P$ is not satisfied for any $\tau \subsetneq \sigma$.  For example, if $U_\sigma = \emptyset$ and for all $\tau \subsetneq \sigma$ we have $U_\tau \neq \emptyset$, then we say that ``$\sigma$ is minimal w.r.t. $U_\sigma = \emptyset$.''

\begin{theorem}\label{thm:canonical-form}
Let $\C \subset \{0,1\}^n$ be a neural code, and let $\U = \{U_1,\ldots,U_n\}$ be any collection of open sets (not necessarily convex) in a nonempty stimulus space $X$ such that $\C = \C(\U)$.   The canonical form of $J_\C$ is:
\begin{eqnarray*}
J_{\C}  &=& \big \langle \big\{x_\sigma \mid \sigma \text{ is minimal w.r.t. } U_\sigma = \emptyset \big\},\\
&& \big\{x_\sigma \prod_{i\in \tau} (1-x_i) \mid \sigma,\tau \neq \emptyset,\; \sigma \cap \tau =\emptyset,\; U_\sigma \neq \emptyset, \; \bigcup_{i \in \tau} U_i \neq X , \text{ and } \sigma, \tau 
\text{ are each minimal }\\ && \text{ w.r.t. } U_\sigma\subseteq \bigcup_{i\in \tau}U_i  \big\},
 \big\{\prod_{i \in \tau} (1-x_i) \mid \tau \text{ is minimal w.r.t. } X \subseteq \bigcup_{i\in \tau} U_i\big\}  \big \rangle.
\end{eqnarray*}
We call the above three (disjoint) sets of relations comprising $CF(J_\C)$ the minimal Type 1 relations, the minimal Type 2 relations, and the minimal Type 3 relations, respectively.
\end{theorem}

\noindent The proof is given in Section~\ref{sec:proof2}.
Note that, because of the uniqueness of the canonical form, if we are given $CF(J_\C)$ then Theorem~\ref{thm:canonical-form} allows us to
read off the corresponding (minimal) relationships that must be satisfied by any receptive field representation 
of the code as $\C = \C(\U)$:

\begin{itemize}
\item Type 1: $x_\sigma \in CF(J_\C)$ implies that $U_\sigma = \emptyset$, but all lower-order intersections $U_\gamma$ with $\gamma \subsetneq \sigma$ are non-empty.
\item Type 2: $x_\sigma \prod_{i\in \tau} (1-x_i)\in CF(J_\C)$ implies that $U_\sigma\subseteq \bigcup_{i\in \tau}U_i$, but
no lower-order intersection is contained in $\bigcup_{i\in \tau}U_i$, and all the $U_i$s are necessary for $U_\sigma\subseteq \bigcup_{i\in \tau}U_i$.
\item Type 3: $\prod_{i \in \tau} (1-x_i) \in CF(J_\C)$ implies that $X \subseteq \bigcup_{i \in \tau} U_i,$ but $X$ is not contained in any lower-order union $\bigcup_{i \in \gamma} U_i$ for $\gamma \subsetneq \tau$.
\end{itemize}
The canonical form $CF(J_\C)$ thus provides a minimal description of the RF structure dictated by the code $\C$.

The Type 1 relations in $CF(J_\C)$ can be used to obtain a (crude) lower bound on the minimal embedding dimension of the neural code, as defined in Section~\ref{sec:beyond}.
Recall Helly's theorem (Section~\ref{sec:helly-nerve}), and observe that if $x_\sigma \in CF(J_\C)$ then $\sigma$ is minimal with respect to $U_\sigma=\emptyset$; this in turn implies that $|\sigma|\leq d+1$.   (If $|\sigma| > d+1$, by minimality all $d+1$ subsets intersect and by Helly's theorem we must have $U_\sigma \neq \emptyset.$)
We can thus obtain a lower bound on the minimal embedding dimension $d$ as 
$$d \geq \max_{\{\sigma \mid x_\sigma \in CF(J_\C)\}} |\sigma|-1,$$ 
where the maximum is taken over all $\sigma$ such that $x_\sigma$ is a Type 1 relation in $CF(J_\C)$.  
This bound only depends on $\Delta(\C)$, however, and
does not provide any insight regarding the different minimal embedding dimensions observed in the examples of Figure 3.  
These codes have no Type 1 relations in their canonical forms, but they are nicely differentiated by their minimal Type 2 and Type 3 relations.
From the receptive field arrangements depicted in Figure 3, we can easily write down $CF(J_\C)$ for each of these codes.
\begin{itemize}
\item[A.] $CF(J_\C) = \{0\}.$  There are no relations here because $\C = 2^{[3]}$.
\item[B.]  $CF(J_\C) = \{1-x_3\}.$  This Type 3 relation reflects the fact that $X = U_3$.
\item[C.] $CF(J_\C) = \{x_1(1-x_2), x_2(1-x_3), x_1(1-x_3)\}.$  These Type 2 relations correspond to $U_1 \subset U_2$, $U_2 \subset U_3$, and $U_1 \subset U_3$.
Note that the first two of these receptive field relationships imply the third; correspondingly, the third canonical form relation satisfies:
$x_1(1-x_3)=(1-x_3)\cdot[x_1(1-x_2)]+x_1\cdot[x_2(1-x_3)].$
\item[D.] $CF(J_\C) = \{(1-x_1)(1-x_2)\}.$ This Type 3 relation reflects $X = U_1 \cup U_2$, and implies $U_3 \subset U_1 \cup U_2$.
\end{itemize}
Nevertheless, we do not yet know how to infer the minimal embedding dimension from $CF(J_\C)$.  
In Appendix 2 (Section~\ref{sec:appendix2}), we provide a complete list of neural codes on three neurons, up to permutation, and their
respective canonical forms.

\subsection{Comparison to the Stanley-Reisner ideal}

Readers familiar with the Stanley-Reisner ideal \cite{MillerSturmfels,StanleyBook} will recognize that this kind of ideal is generated by the Type 1 relations of a neural code $\C$.  The corresponding simplicial complex is $\Delta(\C)$, the smallest simplicial complex that contains the code.

\begin{lemma}  Let $\C = \C(\U)$.  The ideal generated by the Type 1 relations,
$\langle x_\sigma \mid U_\sigma=\emptyset\rangle,$ is the Stanley-Reisner ideal of $\Delta(\C)$.
Moreover, if $\supp\C$ is a simplicial complex, then $CF(J_\C)$ contains no Type 2 or Type 3 relations, and $J_\C$ is thus the Stanley-Reisner ideal for $\supp\C$.
\end{lemma}

\begin{proof}  
To see the first statement, observe that the {\em Stanley-Reisner ideal} of a simplicial complex $\Delta$ is the ideal 
$$I_\Delta \od \langle x_\sigma \mid \sigma \notin \Delta \rangle,$$
and recall that $\Delta(\C)=\{\sigma\subseteq[n] \mid \sigma\subseteq\supp(c)$ for some $c\in \C\}$.  As $\C=\C(\U)$, an equivalent characterization is $\Delta(\C)=\{\sigma\subseteq[n]\mid U_\sigma\neq \emptyset\}$.  
Since these sets are equal, so are their complements in $2^{[n]}$: 
$$\{\sigma\subseteq [n]\mid \sigma\notin\Delta(\C)\}=\{\sigma\subseteq [n] \mid U_\sigma = \emptyset\}.$$
Thus, $\langle x_\sigma\mid U_\sigma=\emptyset\rangle= \langle x_\sigma\mid \sigma\notin\Delta(\C)\rangle$, which is the Stanley-Reisner ideal for $\Delta(\C)$.

\
To prove the second statement, suppose that $\supp\C$ is a simplicial complex.  Note that $\C$ must contain the all-zeros codeword, so  $X \supsetneq \bigcup_{i=1}^n U_i$ and there can be no Type 3 relations.
Suppose the canonical form of $J_{\C}$ contains a Type 2 relation $x_\sigma \prod_{i\in \tau} (1-x_i)$, for some $\sigma,\tau \subset [n]$ satisfying $\sigma, \tau \neq \emptyset$,  $\sigma \cap \tau = \emptyset$ and
 $U_\sigma \neq \emptyset$.  The existence of this relation indicates that $\sigma \notin \supp \C$, while there does exist an $\omega \in \C$ such that $\sigma \subset \omega.$
 This contradicts the assumption that $\supp \C$ is a simplicial complex.  We conclude that $J_{\C}$ has no Type 2 relations.
 \end{proof}

The canonical form of $J_\C$ thus enables us to immediately read off, via the Type 1 relations, the minimal forbidden faces of the simplicial complex $\Delta(\C)$ associated to the code, and also the minimal deviations of $\C$ from being a simplicial complex, which are captured by the Type 2 and Type 3 relations.

\subsection{An algorithm for obtaining the canonical form}
Now that we have established that a minimal description of the RF structure can be extracted from the canonical form of the neural ideal, the most pressing question is the following:
\medskip
 
\noindent{\bf Question:} How do we find the canonical form $CF(J_\C)$ if all we know is the code $\C$, and we are {\em not} given a representation of the code as $\C = \C(\U)$?
\medskip

\noindent In this section we describe an algorithmic method for finding $CF(J_\C)$ from knowledge only of $\C$.
It turns out that computing the primary decomposition of $J_\C$ is a key step towards
finding the minimal pseudo-monomials.  This parallels the situation for monomial ideals, 
although there are some additional subtleties in the case of pseudo-monomial ideals.  As previously discussed,
from the canonical form we can read off the RF structure of the code, so the overall workflow is as follows:
\begin{small}
$$\text{Workflow:} \;\;\; \begin{array}{c}\text{neural code}\\ \C \subset \{0,1\}^n \end{array} \rightarrow 
\begin{array}{c} \text{neural ideal}\\ J_\C = \langle \{\rho_v \mid v \notin \C \} \rangle \end{array} \rightarrow  
\begin{array}{c}\text{primary}\\ \text{decomposition}\\ \text{of}\;J_\C \end{array} \rightarrow 
\begin{array}{c}\text{canonical}\\ \text{form}\\ CF(J_\C) \end{array} \rightarrow 
\begin{array}{c} \text{minimal}\\ \text{RF structure}\\ \text{ of } \C \end{array}$$ 
\end{small}

\noindent {\bf Canonical form algorithm}
\medskip

\noindent{\bf Input:} A neural code $\C \subset \{0,1\}^n$.\medskip

\noindent{\bf Output:} The canonical form of the neural ideal, $CF(J_\C)$.

\begin{itemize}
\item[Step 1:]  From $\C \subset \{0,1\}^n$, compute $J_\C = \big \langle \{\rho_v \mid v \notin \C \} \big\rangle$.
\item[Step 2:]  Compute the primary decomposition of $J_\C$.  It turns out (see Theorem~\ref{thm:prim-decomp} in the next section) that this decomposition yields a unique representation of the ideal as
$$J_\C = \bigcap_{a \in \mathcal{A}} \p_a,$$
where each $a \in \mathcal{A}$ is an element of $\{0,1,*\}^n$, and $\p_a$ is defined as
$$\p_a \od \big\langle \{x_i - a_i \mid a_i \neq *\}\big\rangle =  \big\langle \{x_i \mid a_i = 0\}, \{1-x_j \mid a_j = 1\} \big\rangle.$$
Note that the $\p_a$s are all prime ideals.  We will see later how to compute this primary
decomposition algorithmically, in Section~\ref{sec:prim-decomp-algorithm}.
\item[Step 3:]  Observe that any pseudo-monomial $f \in J_\C$ must satisfy $f \in \p_a$ for each $a \in \mathcal{A}$.
It follows that $f$ is a multiple of one of the linear generators of $\p_a$ for each $a \in \mathcal{A}$.
Compute the following set of elements of $J_\C$:
$$\mathcal{M}(J_\C) = \big\{\prod_{a \in \mathcal{A}} g_a \mid g_a = x_i - a_i \text{ for some } a_i \neq *\big\}.$$
$\mathcal{M}(J_\C)$ consists of all polynomials obtained as a product of linear generators $g_a$, 
one for each prime ideal $\p_a$ of the primary decomposition of $J_\C$.
\item[Step 4:] Reduce the elements of $\mathcal{M}(J_\C)$ by imposing $x_i(1-x_i) = 0$.  This eliminates elements that are not pseudo-monomials.  It also reduces the degrees of some of the remaining elements, as it implies $x_i^2 = x_i$ and $(1-x_i)^2 = (1-x_i)$.  
We are left with a set of pseudo-monomials of the form $f = \prod_{i \in \sigma} x_i \prod_{j\in \tau} (1-x_j)$ for $\tau \cap \sigma = \emptyset.$
Call this new reduced set $\mathcal{\tilde M}(J_\C).$ 
\item[Step 5:] Finally, remove all elements of $\mathcal{\tilde M}(J_\C)$ that are multiples of lower-degree elements in $\mathcal{\tilde M}(J_\C).$ 
\end{itemize}

\begin{proposition}\label{prop:prop1} The resulting set is the canonical form $CF(J_\C)$.
\end{proposition}

\noindent The proof is given in Section~\ref{sec:proof-prop1}.

\subsection{An example}
Now we are ready to use the canonical form algorithm in an example, illustrating how to obtain a possible arrangement of convex receptive fields from a neural code.

Suppose a neural code $\C$ has the following 13 codewords, and 19 missing words: 
\begin{small}
\begin{eqnarray*}
\C&=&\begin{array}{ccccc} \{ 00000, & 10000, &  01000, & 00100, & 00001, \\ 
  11000, & 10001, & 01100, & 00110, & 00101, \\ 
 00011, &  11100, & 00111 \} & & \end{array}\\
  \{0,1\}^5\backslash \C &=& \begin{array}{ccccc}
 \{ 00010, & 10100, & 10010, & 01010, & 01001, \\
 11010, & 11001, & 10110, & 10101, & 10011,\\
 01110, & 01101, & 01011, & 11110, & 11101,\\
 11011, &    10111,& 01111,& 11111 \}.& \\
 \end{array}
\end{eqnarray*}
 \end{small}
 
 Thus, the neural ideal $J_{\C}$ has 19 generators, using the original definition $J_\C = \langle \{ \rho_v \mid v \notin \C\}\rangle$:
 
 \begin{small}
 $$J_\C = \big \langle x_4(1-x_1)(1-x_2)(1-x_3)(1-x_5),
 x_1x_3(1-x_2)(1-x_4)(1-x_5), x_1x_4(1-x_2)(1-x_3)(1-x_5),$$ $$
 x_2x_4(1-x_1)(1-x_3)(1-x_5), x_2x_5(1-x_1)(1-x_3)(1-x_4), x_1x_2x_4(1-x_3)(1-x_5),
 $$ $$x_1x_2x_5(1-x_3)(1-x_4),
 x_1x_3x_4(1-x_2)(1-x_5), x_1x_3x_5(1-x_2)(1-x_4),
 x_1x_4x_5(1-x_2)(1-x_3), $$ $$x_2x_3x_4(1-x_1)(1-x_5),
 x_2x_3x_5(1-x_1)(1-x_4), x_2x_4x_5(1-x_1)(1-x_3),
  x_1x_2x_3x_4(1-x_5)$$ $$x_1x_2x_3x_5(1-x_4), x_1x_2x_4x_5(1-x_3), x_1x_3x_4x_5(1-x_2), 
  x_2x_3x_4x_5(1-x_1), x_1x_2x_3x_4x_5\big \rangle.$$
  \end{small}

\noindent Despite the fact that we are considering only five neurons, this looks like a complicated ideal. Considering the canonical form of $J_\C$ will help us to extract the relevant combinatorial information and allow us to create a possible 
arrangement of receptive fields $\U$ that realizes this code as $\C = \C(\U)$.
Following Step 2 of our canonical form algorithm, we take the primary decomposition of $J_\C$:
 $$J_\C=\langle x_1, x_2, x_4\rangle\cap \langle x_1, x_2, 1-x_3\rangle \cap \langle x_1, x_2, 1-x_5\rangle \cap \langle x_2, x_3, x_4\rangle\cap \langle x_3, x_4, x_5 \rangle \cap \langle x_1, x_4, x_5\rangle \cap \langle 1-x_2, x_4, x_5\rangle.$$
 Then, as described in Steps 3-5 of the algorithm, we take all possible products amongst these seven larger ideals, reducing by the relation $x_i(1-x_i)=0$ (note that this gives us $x_i=x_i^2$ and hence we can say $x_i^k=x_i$ for any $k>1$).  We also remove any polynomials that are multiples of smaller-degree pseudo-monomials in our list.  This process leaves us with six minimal pseudo-monomials, yielding the canonical form:  
$$J_{\C} = \langle CF(J_\C) \rangle = \langle x_1x_3x_5,\, x_2x_5,\, x_1x_4,\, x_2x_4,\, x_1x_3(1-x_2),\, x_4(1-x_3)(1-x_5)\rangle.$$
Note in particular that every  generator we originally put
  in $J_{\mathcal C}$ is a multiple of one of the six
  relations in $CF(J_\C)$.  Next, we consider what the relations in $CF(J_\C)$ tell us 
  about the arrangement of receptive fields that would be needed to realize the code as $\C = \C(\U)$.

\begin{enumerate}

\item $x_1x_3x_5 \in CF(J_\C) \Rightarrow U_1\cap U_3\cap U_5=\emptyset$, while $U_1\cap U_3, U_3\cap U_5$ and $U_1\cap U_5$ are all nonempty.

\item $x_2x_5 \in CF(J_\C) \Rightarrow U_2\cap U_5=\emptyset$, while $U_2, U_5$ are both nonempty.  

\item $x_1x_4 \in CF(J_\C) \Rightarrow U_1\cap U_4=\emptyset$, while $U_1, U_4$ are both nonempty.

\item $x_2x_4 \in CF(J_\C) \Rightarrow U_2\cap U_4=\emptyset$, while $U_2, U_4$ are both nonempty.

\item $x_1x_3(1-x_2)\in CF(J_\C) \Rightarrow U_1\cap U_3\subseteq U_2$, while $U_1\not\subseteq U_2, U_3\not\subseteq U_2$, and $U_1\cap U_3\neq \emptyset$.

\item $x_4(1-x_3)(1-x_5)\in CF(J_\C) \Rightarrow U_4\subseteq U_3\cup U_5$, while $U_4\neq \emptyset$, and that $U_4\not\subseteq U_3, U_4\not\subseteq U_5$.  

\end{enumerate}

The minimal Type 1 relations (1-4) tell us that we should draw $U_1, U_3$ and $U_5$ with all pairwise intersections, 
but leaving a ``hole'' in the middle since the triple intersection is empty.  Then $U_2$ should be drawn to intersect $U_1$ and $U_3$, but not $U_5$.
Similarly, $U_4$ should intersect $U_3$ and $U_5$, but not $U_1$ or $U_2$.
The minimal Type 2 relations (5-6) tell us that $U_2$ should be drawn to contain the intersection $U_1 \cap U_3$, while $U_4$ lies in the union $U_3 \cup U_5$, but
is not contained in $U_3$ or $U_5$ alone.  There are no minimal Type 3 relations, as expected for a code that includes the all-zeros codeword.

Putting all this together, and assuming convex receptive fields, we can completely infer the receptive field structure, and draw
the corresponding picture (see Figure 4).
It is easy to verify that the code $\C(\U)$ of the pictured arrangement indeed coincides with $\C$.  

\begin{figure}[h]
\centering
   \vspace{-.1in}
   \includegraphics[width=2.5in]{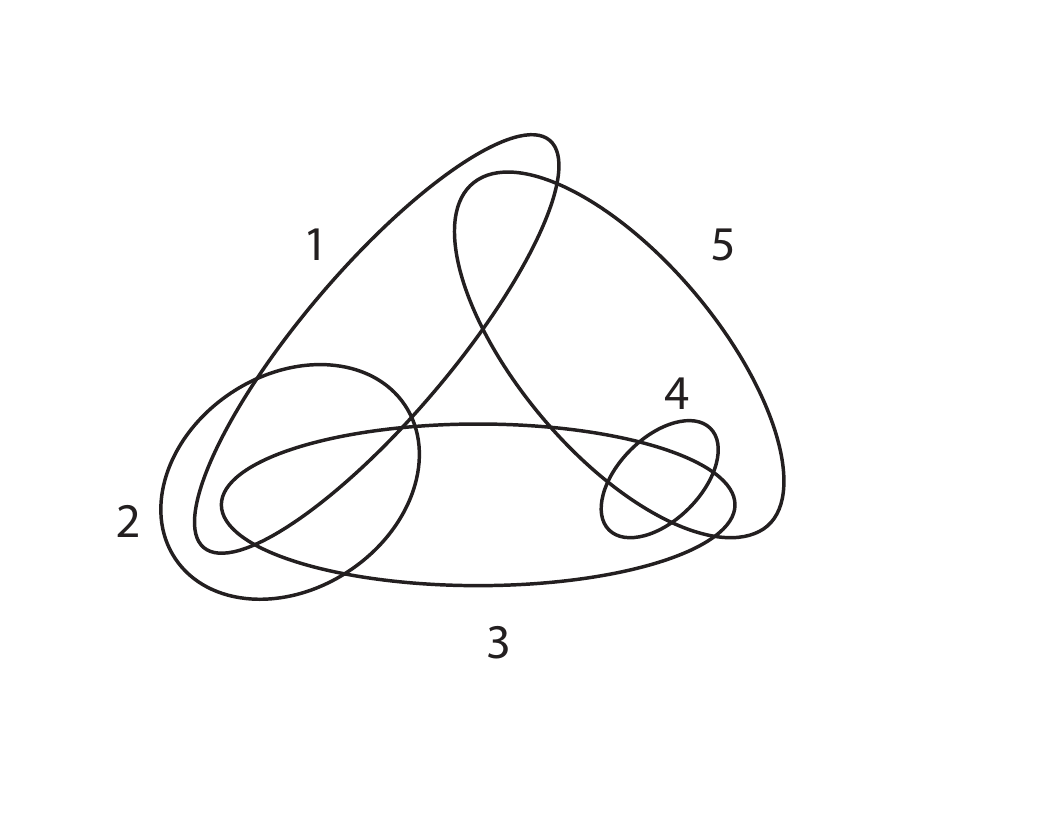}    
    \vspace{-.3in}
   \caption{\small An arrangement of five sets that realizes $\C$ as $\C(\U)$. }
 \end{figure}

\section{Primary decomposition}\label{sec:prim-decomp}

Let $\C \subset \{0,1\}^n$ be a neural code.  The primary decomposition of $I_\C$ is boring: $$I_\C = \bigcap_{c \in \C} m_c,$$ 
where $m_v$ for any $v \in \{0,1\}^n$ is the maximal ideal $I(v)$ defined in Section~\ref{sec:spec}.
This simply expresses $I_\C$ as the intersection of all maximal ideals $m_c$ for $c \in \C$, because
the variety $\C = V(I_\C)$ is just a finite set of points and the primary decomposition reflects no additional structure of the code.

On the other hand, the primary decomposition of the neural ideal $J_\C$ retains the full combinatorial structure of $\C$.  Indeed, we have seen that computing this decomposition is a critical step towards obtaining $CF(J_\C)$, which captures the receptive field structure of the neural code.  In this section, we describe the primary decomposition of $J_\C$ and discuss its relationship to some natural decompositions of the neural code.  We end with an algorithm for obtaining primary decomposition of any pseudo-monomial ideal.

\subsection{Primary decomposition of the neural ideal}

We begin by defining some objects related to $\F_2[x_1,\ldots,x_n]$ and $\{0,1\}^n$, without reference to any particular neural code.
For any $a \in \{0,1,*\}^n$, we define the variety
$$V_a \od \{v \in \{0,1\}^n \mid v_i = a_i \text{ for all } i \text{ s.t. } a_i \neq *\} \subseteq \{0,1\}^n.$$
This is simply the subset of points compatible with the word ``$a$'', where $*$ is viewed as a ``wild card'' symbol.  Note that $V_v = \{v\}$ for any $v \in \{0,1\}^n$.
We can also associate a prime ideal to $a$,
$$\p_a \od \langle \{x_i - a_i \mid a_i \neq *\}\rangle \subseteq \F_2[x_1,\ldots,x_n],$$
consisting of polynomials in $\F_2[x_1,\ldots,x_n]$ that vanish on all points compatible with $a$.  To obtain all such polynomials,
we must add in the Boolean relations (see Section~\ref{sec:lemma-proofs}):
$$\q_a \od I(V_a) = \p_a + \langle x_1^2-x_1,\ldots,x_n^2-x_n \rangle.$$
Note that $V_a = V(\p_a) = V(\q_a)$.

Next, let's relate this all to a code $\C \subset \{0,1\}^n$. Recall the definition of the neural ideal,
$$J_\C \od  \langle \{ \rho_v \mid v \notin \C\}\rangle = \langle \{ \prod_{i=1}^n((x_i-v_i)-1) \mid v \notin \C\}\rangle.$$
We have the following correspondences.

 \begin{lemma}\label{lemma:lemma6}  
 $J_\C \subseteq \p_a \Leftrightarrow V_a \subseteq \C.$
 \end{lemma}

 \begin{proof}
 ($\Rightarrow$)  $J_\C \subseteq \p_a \Rightarrow V(\p_a) \subseteq V(J_\C).$  Recalling that  $V(\p_a) = V_a$ and   $V(J_\C) = \C$, this gives $V_a \subseteq \C.$\\
($\Leftarrow$) $V_a \subseteq \C \Rightarrow I(\C) \subseteq I(V_a) \Rightarrow I_\C \subseteq \q_a.$  Recalling that both $I_\C$ and $\q_a$ differ
from $J_\C$ and $\p_a$, respectively, by the addition of the Boolean relations, we obtain $J_\C \subseteq \p_a$.
 \end{proof}

\begin{lemma}\label{lemma:lemma7}
 For any $a,b\in\{0,1,*\}^n$,
$V_a\subseteq V_b \Leftrightarrow \mathbf{p}_b\subseteq\mathbf{p}_a.$
\end{lemma}

\begin{proof}
($\Rightarrow$) Suppose $V_a\subseteq V_b$. Then, for any $i$ such that $b_i\neq *$ we have $a_i=b_i$. It follows that each generator of $\mathbf{p}_b$ is also in $\mathbf{p}_a$, so $\p_b \subseteq \p_a$.
($\Leftarrow$) Suppose $\mathbf{p}_b\subseteq\mathbf{p}_a$. Then, $V_a=V(\mathbf{p}_a)\subseteq V(\mathbf{p}_b)=V_b.$
\end{proof}

Recall that a an ideal $\p$ is said to be a {\em minimal prime} over $J$ if $\p$ is a prime ideal that contains $J$, and there is no other prime ideal $\p'$ such that $\p \supsetneq \p' \supseteq J$.  Minimal primes $\p_a \supseteq J_\C$ correspond to maximal varieties $V_a$ such that $V_a \subseteq \C$.
Consider the set
 $$\A_\C \od \{ a \in \{0,1,*\}^n \mid V_a \subseteq \C\}.$$
We say that $a \in \A_\C$ is {\em maximal} if there does not exist another element $b \in \A_\C$ such that $V_a \subsetneq V_b$ (i.e., $a \in \A_\C$ is maximal if $V_a$ is maximal such that $V_a \subseteq \C$).

\begin{lemma}\label{lemma:correspondence}
The element $a \in \A_\C$ is maximal if and only if $\p_a$ is a minimal prime over $J_\C$.
\end{lemma}

\begin{proof}
Recall that $a \in \A_\C \Rightarrow V_a \subseteq \C$, and hence $J_\C \subseteq \p_a$ (by Lemma~\ref{lemma:lemma6}).
($\Rightarrow$) Let $a \in \A_\C$ be maximal, and choose $b \in \{0,1,*\}$ such that 
$J_\C \subseteq \p_b \subseteq \p_a$.
By Lemmas~\ref{lemma:lemma6} and~\ref{lemma:lemma7}, $V_a \subseteq V_b \subseteq \C$.  Since $a$ is maximal, we conclude that $b = a$, and hence $\p_b = \p_a$.  
It follows that $\p_a$ is a minimal prime over $J_\C$.  ($\Leftarrow$) Suppose $\p_a$ is a minimal prime over $J_\C$.  Then by Lemma~\ref{lemma:lemma6}, $a \in \A_\C$.  Let
$b$ be a maximal element of $\A_\C$ such that $V_a \subseteq V_b \subseteq \C$.  Then $J_\C \subseteq \p_b \subseteq \p_a$.  Since $\p_a$ is a minimal prime over $J_\C$, $\p_b = \p_a$ and hence $b = a$.  Thus $a$ is maximal in $\A_\C$.
\end{proof}

\noindent We can now describe the primary decomposition of $J_\C$.  Here we assume the neural code $\C \subseteq \{0,1\}^n$ is non-empty, so that $J_\C$ is a proper pseudo-monomial ideal.

\begin{theorem}\label{thm:prim-decomp}
$J_\C = \bigcap_{i=1}^\ell \p_{a_i}$ is the unique irredundant primary decomposition of $J_\C$, where $\p_{a_1},\ldots,\p_{a_\ell}$ are the minimal primes over $J_\C$. 
\end{theorem}

\noindent The proof is given in Section~\ref{sec:prim-decomp-proof}.  Combining this theorem
with Lemma~\ref{lemma:correspondence}, we have:

\begin{corollary}
$J_\C = \bigcap_{i=1}^\ell \p_{a_i}$ is the unique irredundant primary decomposition of $J_\C$, where 
$a_1,\ldots,a_\ell$ are the maximal elements of $A_\C$.
\end{corollary}

\subsection{Decomposing the neural code via intervals of the Boolean lattice}\label{sec:boolean-lattice}

From the definition of $\A_\C$, it is easy to see that the maximal elements yield a kind of ``primary'' decomposition
of the neural code $\C$ as a union of maximal $V_a$s.

\begin{lemma} \label{lemma:code-decomp}
$\C = \bigcup_{i=1}^\ell V_{a_i}$, where $a_1,\ldots,a_\ell$ are the maximal elements of $\A_\C$.  (I.e., $\p_{a_1},\ldots,\p_{a_\ell}$
are the minimal primes in the primary decomposition of $J_\C$.)
\end{lemma}

\begin{proof}
Since $V_a \subseteq \C$ for any $a \in \A_\C$, clearly $\bigcup_{i=1}^\ell V_{a_i} \subseteq \C$.
To see the reverse inclusion, note that for any $c \in \C$, $c \in V_c \subseteq V_a$ for some maximal $a \in \A_\C$.
Hence, $\C \subseteq \bigcup_{i=1}^\ell V_{a_i}.$
\end{proof}

Note that Lemma~\ref{lemma:code-decomp}
could also be regarded as a corollary of Theorem~\ref{thm:prim-decomp}, since $\C = V(J_\C) = V(\bigcap_{i=1}^\ell \p_{a_i}) = \bigcup_{i=1}^\ell V(\p_{a_i}) = \bigcup_{i=1}^\ell V_{a_i}$, and the maximal $a \in \A_\C$ correspond to minimal primes $\p_a \supseteq J_\C$.
Although we were able to prove Lemma~\ref{lemma:code-decomp} directly, in practice we use the primary decomposition in order to find (algorithmically) the maximal elements $a_1,\ldots,a_\ell \in \A_\C$, and thus determine the $V_a$s for the above decomposition of the code.

It is worth noting here that the decomposition of $\C$ in Lemma~\ref{lemma:code-decomp} is not necessarily {\it minimal}.  This is because one can have fewer $\q_a$s such that
$$\bigcap_{i \in \sigma \subsetneq [\ell]} \q_{a_i} = \bigcap_{i \in [\ell]} \p_{a_i}.$$
Since $V(\q_{a_i}) = V(\p_{a_i}) = V_{a_i}$, this would lead to a decomposition of $\C$ as a union of fewer $V_{a_i}$s.  In contrast, the primary decomposition of $J_\C$ in Theorem~\ref{thm:prim-decomp} is irredundant, and hence
none of the minimal primes can be dropped from the intersection.

\subsubsection*{Neural activity ``motifs'' and intervals of the Boolean lattice}

We can think of an element $a \in \{0,1,*\}^n$ as a neural activity ``motif''.  That is, $a$ is a pattern of activity and silence for a subset of the neurons,
while $V_a$ consists of all activity patterns on the full population of neurons that are consistent with this motif (irrespective of what the code is).
For a given neural code $\C$, the set of maximal $a_1,\ldots,a_l \in \A_\C$ corresponds to a set of minimal motifs that define the code (here ``minimal'' is used in the sense of having the fewest number of neurons that are constrained to be ``on'' or ``off'' because $a_i \neq *$).
If $a \in \{0,*\}^n$, we refer to $a$ as a neural {\em silence} motif, since it corresponds to a pattern of silence.  In particular, silence motifs correspond to simplices in $\supp \C$,
since $\supp V_a$ is a simplex in this case.  If $\supp \C$ is a simplicial complex, then Lemma~\ref{lemma:code-decomp} gives the decomposition of $\C$ as a union of minimal silence motifs (corresponding to {\it facets}, or maximal simplices, of $\supp \C$).

More generally, $V_a$ corresponds to an {\em interval} of the Boolean lattice $\{0,1\}^n$.  Recall the poset structure of the Boolean lattice: for any pair of elements $v_1,v_2 \in \{0,1\}^n$, we have $v_1 \leq v_2$ if and only if $\supp(v_1) \subseteq \supp(v_2)$.  An {\it interval} of the Boolean lattice is thus a subset of the form:
$$[u_1,u_2] \od \{ v \in \{0,1\}^n \mid u_1 \leq v \leq u_2 \}.$$
Given an element $a \in \{0,1,*\}^n$, we have a natural interval consisting of all Boolean lattice elements ``compatible'' with $a$.  Letting $a^0 \in \{0,1\}^n$ be the element
obtained from $a$ by setting all $*$s to $0$, and $a^1 \in \{0,1\}^n$ the element obtained by setting all $*$s to $1$, we find that
$$V_a = [a^0,a^1] = \{ v \in \{0,1\}^n \mid a^0 \leq v \leq a^1 \}.$$
Simplices correspond to intervals of the
form $[0,a^1]$, where $0$ is the bottom ``all-zeros'' element in the Boolean lattice.

While the primary decomposition of $J_\C$ allows a neural code $\C \subseteq \{0,1\}^n$ to be decomposed as a union of intervals of the Boolean lattice,
as indicated by Lemma~\ref{lemma:code-decomp}, the canonical form $CF(J_\C)$ provides a decomposition of the {\it complement} of $\C$ as a union of intervals.
First, notice that to any pseudo-monomial $f \in CF(J_\C)$ we can associate an element $b \in \{0,1,*\}$ as follows: $b_i = 1$ if $x_i | f$, $b_i = 0$ if $(1-x_i) | f$, and $b_i = *$ otherwise.  In other words,
$$f = f_b \od \prod_{\{i \mid b_i = 1\}} x_i  \prod_{\{j \mid b_j = 0\}} (1-x_j).$$
As before, $b$ corresponds to an interval $V_b = [b^0,b^1]\subset \{0,1\}^n$.  Recalling the $J_\C$ is generated by pseudo-monomials corresponding to non-codewords,
it is now easy to see that the complement of $\C$ in $\{0,1\}^n$ can be expressed as the union of $V_b$s, where each $b$ corresponds to a pseudo-monomial in the canonical form.   The canonical form thus provides an alternative description of the code, nicely complementing Lemma~\ref{lemma:code-decomp}.

\begin{lemma}
$\C = \{0,1\}^n \setminus \bigcup_{i=1}^k V_{b_i}$, where $CF(J_\C) = \{f_{b_1},\ldots,f_{b_k}\}$.
\end{lemma}

\begin{wrapfigure}{r}{.5\linewidth}
\vspace{-.35in}
   \centering
   \includegraphics[width=1.8in]{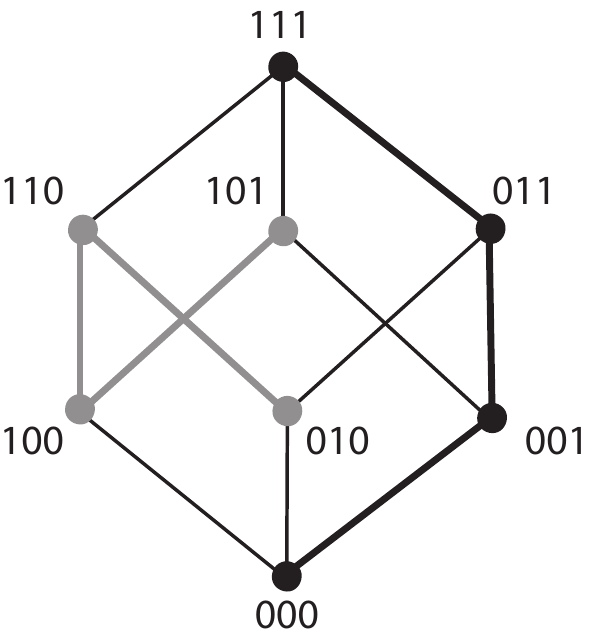} 
   \vspace{-.1in}
   \caption{\small Boolean interval decompositions of the code $\C = \{000, 001, 011, 111\}$ (in black) and of its complement (in gray), arising
   from the primary decomposition and canonical form of $J_\C$, respectively.}
\end{wrapfigure}

\noindent We now illustrate both decompositions of the neural code with an example.
\medskip

\noindent{\bf Example.} Consider the neural code $\C = \{000, 001, 011, 111\} \subset \{0,1\}^3$ corresponding to a set of receptive fields satisfying $U_1 \subsetneq U_2 \subsetneq U_3 \subsetneq X$.  The primary decomposition of $J_\C \subset \F_2[x_1,x_2,x_3]$ is given by 
$$\langle x_1,x_2 \rangle \cap \langle x_1,1-x_3 \rangle \cap \langle 1-x_2, 1-x_3 \rangle,$$
while the canonical form is
$$CF(J_\C) = \langle x_1(1-x_2), x_2(1-x_3), x_1(1-x_3) \rangle.$$

From the primary decomposition, we can write $\C = V_{a_1} \cup V_{a_2} \cup V_{a_3}$ for $a_1 = 00*$, $a_2 = 0{*}1$, and $a_3 = *11$.
The corresponding Boolean lattice intervals are $[000,001]$, $[001,011]$, and $[011,111]$, respectively, and are depicted in black in Figure 5.
As noted before, this decomposition of the neural code need not be minimal; indeed, we could also write $\C = V_{a_1} \cup V_{a_3}$, as the 
middle interval is not necessary to cover all codewords in $\C$.

From the canonical form, we obtain $\C = \{0,1\}^3 \setminus (V_{b_1} \cup V_{b_2} \cup V_{b_3})$, where $b_1 = 10*$, $b_2 = *10$, and $b_3=1{*}0.$
The corresponding Boolean lattice intervals spanning the complement of $\C$ are $[100,101]$, $[010,110]$, and $[100,110]$, respectively; these are 
depicted in gray in Figure 5.  Again, notice that this decomposition is not minimal -- namely, $V_{b_3} = [100,110]$ could be dropped.

\subsection{An algorithm for primary decomposition of pseudo-monomial ideals}\label{sec:prim-decomp-algorithm}

We have already seen that computing the primary decomposition of the neural ideal $J_\C$ is a critical step towards extracting the canonical form $CF(J_\C)$, and that it also yields a meaningful decomposition of $\C$ in terms of neural activity motifs.
Recall from  Section~\ref{sec:canonical-form} that $J_\C$ is always a {\it pseudo-monomial ideal} -- i.e., $J_\C$ is generated by pseudo-monomials, which are polynomials $f \in \F_2[x_1,\ldots,x_n]$ of the form
$$f = \prod_{i \in \sigma} z_i, \;\; \text{where} \;\; z_i \in \{x_i, 1-x_i\} \;\; \text{for any} \;\; i \in [n].$$
In this section, we provide an explicit algorithm for finding
the primary decomposition of such ideals.

In the case of {\it monomial} ideals, there are many algorithms for obtaining the primary decomposition, and there are already fast implementations of such algorithms in algebraic geometry software packages such as Singular and Macaulay2 \cite{macaulay-book}.  
Pseudo-monomial ideals are closely related to square-free monomial ideals, but there are some differences which require a bit of care.  
In particular, if $J \subseteq F_2[x_1,\ldots,x_n]$ is a pseudo-monomial ideal  and $z \in \{x_i, 1-x_i\}$ for some $i \in [n]$, then
for $f$ a pseudo-monomial:
$$f \in \langle J, z \rangle \not\Rightarrow f \in J \text{ or } f \in \langle z \rangle.$$
To see why, observe that $x_1 \in \langle x_1(1-x_2), x_2 \rangle$, because
$x_1 = 1 \cdot x_1(1-x_2) + x_1 \cdot x_2,$ but $x_1$ is not a multiple of either $x_1(1-x_2)$ or $x_2$.
We can nevertheless adapt ideas from (square-free) monomial ideals to obtain an
algorithm for the primary decomposition of pseudo-monomial ideals.  The following lemma allows us to handle the above complication.

\begin{lemma}\label{new-lemma}
Let $J\subset \F_2[x_1,\ldots,x_n]$ be a pseudo-monomial ideal, and let $z \in \{x_i, 1-x_i\}$ for some $i \in [n]$.  For any pseudo-monomial $f$,
$$f \in \langle J,z \rangle \Rightarrow f \in J \text{ or } f \in \langle z \rangle \text{ or } (1-z)f \in J.$$
\end{lemma}

\noindent The proof is given in Section~\ref{sec:prop2-proof}.
Using Lemma~\ref{new-lemma} we can prove the following key lemma for our algorithm, which mimics the case of square-free monomial ideals.

\begin{lemma} \label{lemma:z-decomp}
Let $J \subset \F_2[x_1,\ldots,x_n]$ be a pseudo-monomial ideal, and let
$\prod_{i \in \sigma} z_i$ be a pseudo-monomial, with $z_i \in \{x_i, 1-x_i\}$ for each $i$.  Then,
$$\langle J, \prod_{i \in \sigma} z_i \rangle = \bigcap_{i \in \sigma} \langle J, z_i \rangle.$$
\end{lemma}

\noindent The proof is given in Section~\ref{sec:prop2-proof}.  Note that if $\prod_{i \in \sigma} z_i \in J$, then this lemma implies  $J =  \bigcap_{i \in \sigma} \langle J, z_i \rangle,$
which is the key fact we will use in our algorithm.
This is similar to Lemma 2.1 in \cite[Monomial Ideals Chapter]{macaulay-book},
and suggests a recursive algorithm along similar lines to those that exist for monomial ideals.

The following observation will add considerable efficiency to our algorithm for pseudo-monomial ideals.

\begin{lemma}\label{lemma:x_i-reduction}
Let $J \subset \F_2[x_1,\ldots,x_n]$ be a pseudo-monomial ideal.  For any $z_i \in \{x_i,1-x_i\}$ we can write
$$J = \langle z_i g_1, \ldots ,z_i g_k, (1-z_i)f_1,\ldots,(1-z_i)f_\ell,h_1,\ldots,h_m\rangle,$$
where the $g_j$, $f_j$ and $h_j$ are pseudo-monomials that contain no $z_i$ or $1-z_i$ term.  (Note that
$k, \ell$ or $m$ may be zero if there are no generators of the corresponding type.)
Then, 
\begin{eqnarray*}
\langle J, z_i \rangle = \langle J|_{z_i=0}, z_i \rangle &=& \langle z_i, f_1,\ldots,f_\ell,h_1,\ldots,h_m \rangle.
\end{eqnarray*}
\end{lemma}

\begin{proof}
Clearly, the addition of $z_i$ in $\langle J,z_i \rangle$ renders the $z_ig_j$ generators unnecessary.
The $(1-z_i)f_j$ generators can be reduced to just $f_j$ because $f_j = 1\cdot(1-z_i)f_j + f_j \cdot z_i$.
\end{proof}

\noindent We can now state our algorithm.   Recall that an ideal $I \subseteq R$ is {\em proper} if $I \neq R$.

\subsubsection*{Algorithm for primary decomposition of pseudo-monomial ideals}

\noindent {\bf Input}: A proper pseudo-monomial ideal $J \subset \F_2[x_1,\ldots,x_n]$.  This is presented as $J = \langle g_1,\ldots,g_r\rangle$ with each generator $g_i$ a pseudo-monomial.\\

\noindent {\bf Output}: Primary decomposition of $J$.  This is returned as a set $\P$ of prime ideals, with $J = \bigcap_{I \in \P} I$.

\begin{itemize}
\item Step 1 (Initializion Step): Set $\P = \emptyset$ and $D = \{J\}.$  Eliminate from the list of generators of $J$ those that are multiples of other generators.

\item Step 2 (Splitting Step): For each ideal $I \in D$ compute $D_I$ as follows.   
\begin{itemize}
\item[Step 2.1:] Choose a nonlinear generator $z_{i_1}\cdots z_{i_m} \in I$, 
where each $z_i \in \{x_i, 1-x_i\}$, and $m \geq 2$.  
(Note: the generators of $I$ should always be pseudo-monomials.)
\item[Step 2.2:] Set $D_I = \{ \langle I, z_{i_1} \rangle, \ldots, \langle I, z_{i_m} \rangle\}.$
By Lemma~\ref{lemma:z-decomp} we know that
$$I = \bigcap_{k=1}^m \langle I, z_{i_k} \rangle = \bigcap_{K \in D_I} K.$$
\end{itemize}

\item Step 3 (Reduction Step):  For each $D_I$ and each ideal  $\langle I, z_i\rangle \in D_I$, reduce the set of generators as follows.
\begin{itemize}
\item[Step 3.1:] Set $z_i = 0$ in each generator of $I$.  This yields a ``0'' for each multiple of $z_i$, and removes $1-z_i$ factors in each of the remaining generators.  
By Lemma~\ref{lemma:x_i-reduction}, $\langle I,z_i \rangle = \langle I |_{z_i = 0}, z_i \rangle$. 
\item[Step 3.2:] Eliminate $0$s and generators that are multiples of other generators.
\item[Step 3.3:] If there is a $``1"$ as a generator, eliminate $\langle I, z_i \rangle$ from $D_I$ as it is not a proper ideal.
\end{itemize}
 
 \item Step 4 (Update Step): Update $D$ and $\P$, as follows.
 \begin{itemize}
 \item[Step 4.1:] Set $D = \bigcup D_I$, and remove redundant ideals in $D$.  That is, remove an ideal if it has the same set of generators as another ideal in $D$.
 \item [Step 4.2:] For each ideal $I \in D$, if $I$ has only linear generators (and is thus prime), move $I$ to $\P$ by setting  $\P = \P \cup I$ and $D = D \setminus I$.
\end{itemize}

 \item Step 5 (Recursion Step): Repeat Steps 2-4 until $D = \emptyset$.
 
 \item Step 6 (Final Step): Remove redundant ideals of $\P$.   That is, remove ideals that are not necessary to preserve the equality $J = \bigcap_{I \in \P} I$.
 \end{itemize}

\begin{proposition}\label{prop:prim-decomp}
This algorithm is guaranteed to terminate, and the final $\P$ is a set of irredundant prime ideals such that $J = \bigcap_{I \in \P} I$.
\end{proposition}

\begin{proof}
For any pseudo-monomial ideal $I \in D$, let $\deg(I)$ be the sum of the degrees of all generating monomials of $I$.  
To see that the algorithm terminates, observe that for each ideal $\langle I, z_i \rangle \in D_I$,  $\deg(\langle I, z_i \rangle) < \deg(I)$ (this follows from
Lemma~\ref{lemma:x_i-reduction}).  The degrees of elements in $D$ thus steadily decrease with each recursive iteration, until they are removed as prime ideals
that are appended to $\P$.  At the same time, the size of $D$ is strictly bounded at $|D| \leq 2^{n \choose 3}$, since there are only $n \choose 3$ pseudo-monomials in $\F_2[x_1,\ldots,x_n]$,
and thus at most $2^{n \choose 3}$ distinct pseudo-monomial ideals.

By construction, the final $\P$ is an irredundant set of prime ideals.  Throughout the algorithm, however, it is always true that 
$J = \left(\bigcap_{I \in D} I\right) \cap \left(\bigcap_{I \in \P} I \right)$.  Since the final $D = \emptyset$, the final $\P$ satisfies $J = \bigcap_{I \in \P} I$.
\end{proof}

\section*{Acknowledgments}
CC was supported by NSF DMS 0920845 and NSF DMS 1225666, a Woodrow Wilson Career Enhancement
Fellowship, and an Alfred P. Sloan Research Fellowship.  VI was supported by NSF DMS 0967377, NSF DMS 1122519, and the Swartz Foundation.

\section{Appendix 1: Proofs}

\subsection{Proof of Lemmas~\ref{lemma:spec} and~\ref{lemma:explicit-relations}}\label{sec:lemma-proofs}

To prove Lemmas~\ref{lemma:spec} and~\ref{lemma:explicit-relations}, we need a version of the Nullstellensatz for finite fields.
The original ``Hilbert's Nullstellensatz'' applies when $k$ is an algebraically closed field. It states that if $f \in k[x_1,\ldots,x_n]$ vanishes on $V(J)$, then $f \in \sqrt{J}$.  In other words,
$$I(V(J)) = \sqrt{J}.$$
Because we have chosen $k = \F_2 = \{0,1\}$, we have to be a little careful about the usual ideal-variety correspondence, as there are some subtleties introduced in the case of finite fields.   In particular, $J = \sqrt{J}$ in $\F_2[x_1,\ldots,x_n]$ does not imply $I(V(J)) = J$.

The following lemma and theorem are well-known.
Let $\F_q$ be a finite field of size $q$, and $\F_q[x_1,\ldots,x_n]$ the $n$-variate polynomial ring
over $\F_q$.

\begin{lemma}\label{lemma:radical}
For any ideal $J \subseteq \F_q[x_1,\ldots,x_n]$, the ideal $J+\langle x_1^q-x_1, \ldots, x_n^q-x_n \rangle$ is a radical ideal.
\end{lemma}

\begin{theorem}[Strong Nullstellensatz in Finite Fields]  For an arbitrary finite field $\F_q$, let $J \subseteq \F_q[x_1,\ldots,x_n]$ be an ideal.  Then,
$$I(V(J)) = J + \langle x_1^q-x_1, \ldots, x_n^q-x_n \rangle.$$
\end{theorem}


\subsubsection*{Proof of Lemma~\ref{lemma:spec}}

 We begin by describing the maximal ideals of $\F_2[x_1,\ldots,x_n]$.
Recall that
$$m_v \od I(v) =  \{f \in \F_2[x_1,\ldots,x_n] \mid f(v) = 0 \}$$
is the maximal ideal of $\F_2[x_1,\ldots,x_n]$ consisting of all functions that vanish on $v \in \F_2^n$.  We will use the notation $\bar{m}_v$
to denote the quotient of $m_v$ in $R_\C$, in cases where $m_v \supset I_\C$.

\begin{lemma}  \label{lemma:mv}
$m_v = \langle x_1-v_1, \ldots, x_n-v_n \rangle \subset \F_2[x_1,\ldots,x_n]$, and is a radical ideal.
\end{lemma}
\begin{proof}  Denote $A_v = \langle x_1-v_1, \ldots, x_n-v_n \rangle$, and observe that $V(A_v) = \{v\}$.
It follows that $I(V(A_v)) = I(v) =  m_v$.  On the other hand, using the Strong Nullstellensatz in Finite Fields we have
$$I(V(A_v)) = A_v + \langle x_1^2-x_1,\ldots,x_n^2-x_n \rangle = A_v,$$
where the last equality is obtained by observing that, since $v_i \in \{0,1\}$ and $x_i^2-x_i = x_i(1-x_i)$, each generator of $ \langle x_1^2-x_1,\ldots,x_n^2-x_n \rangle$ is already contained in $A_v$.
We conclude that $A_v = m_v$, and the ideal is radical by Lemma~\ref{lemma:radical}.
\end{proof}

\noindent In the proof of Lemma~\ref{lemma:spec}, we make use of the following correspondence: for any quotient ring $R/I$, the maximal ideals of $R/I$ are exactly the quotients $\bar m = m/I$, where $m$ is a maximal ideal of $R$ that contains $I$  \cite{atiyah-macdonald}.

\begin{proof}[Proof of Lemma~\ref{lemma:spec}]
First, recall that because $R_\C$ is a Boolean ring, $\mathrm{Spec}(R_\C) = \mathrm{maxSpec}(R_\C)$, the set of all maximal ideals of $R_\C$.  We
also know that the maximal ideals of $\F_2[x_1,\ldots,x_n]$ are exactly those of the form $m_v$ for $v\in \F_2^n$.  By the correspondence stated above, to show that
$\mathrm{maxSpec}(R_\C)  = \{ \bar{m}_v \mid v \in \C\}$ it suffices to show $m_v \supset I_\C$ if and only if $v\in \C$. To see this, note that for each $v\in \C$, $I_\C\subseteq m_v$ because, by definition, all elements of $I_\C$ are functions that vanish on each $v \in \C$. On the other hand, if $v\notin \C$ then $m_v \not\supseteq I_\C$; in particular,
the characteristic function $\rho_v \in I_\C$ for $v \notin \C$, but  $\rho_v \notin m_v$ because $\rho_v(v) = 1$.
Hence, the maximal ideals of $R_\C$ are exactly those of the form $\bar m_v$ for $v\in \C$.
\end{proof} 

We have thus verified that the points in $\mathrm{Spec}(R_\C)$ correspond to codewords in $\C$.  This was expected given our original
definition of the neural ring, and  suggests that the relations on $\F_2[x_1,\ldots,x_n]$ imposed by $I_\C$ are simply relations ensuring that 
$V(\bar m_v) = \emptyset$ for all $v \notin \C$.  

\subsubsection*{Proof of Lemma~\ref{lemma:explicit-relations}}
Here we find explicit relations for $I_\C$ in the case of an arbitrary neural code. 
Recall that
$$\rho_v = \prod_{i=1}^n((x_i-v_i)-1) = \prod_{\{i\,|\,v_i=1\}}x_i\prod_{\{j\,|\,v_j=0\}}(1-x_j),$$
and that $\rho_v(x)$ can be thought of as a characteristic function for $v$, since it satisfies $\rho_v(v) = 1$ and $\rho_v(x) = 0$ for any other
$x \in \F_2^n$.  This immediately implies that
$$V(J_\C) = V(\langle \{ \rho_v \mid v \notin \C\}\rangle) = \C.$$
We can now prove Lemma~\ref{lemma:explicit-relations}.
\begin{proof}[Proof of Lemma~\ref{lemma:explicit-relations}]
Observe that $I_\C = I(\C) = I(V(J_\C))$, since $V(J_\C) = \C$.  On the other hand, the
Strong Nullstellensatz  in Finite Fields implies $I(V(J_\C)) = J_\C + \langle x_1^2-x_1,\ldots,x_n^2-x_n \rangle = J_\C + \B.$
\end{proof}

\subsection{Proof of Theorem~\ref{thm:ideal-equivalence}}\label{sec:proof1}

Recall that for a given set of receptive fields $\U = \{U_1,\ldots,U_n\}$ in some stimulus space $X$, the ideal $I_\U \subset \F_2[x_1,\ldots,x_n]$ was 
defined as:
$$I_\U \od \big \langle \{ x_\sigma \prod_{i \in \tau} (1-x_i) \mid U_\sigma \subseteq
\bigcup_{i \in \tau}U_i \} \big \rangle .$$
The Boolean relations are present in $I_\U$ irrespective of $\U$, as it is always true that $U_i \subseteq U_i$ and this yields the relation $x_i(1-x_i)$ for each $i$.  By analogy with our definition of $J_\C$, it makes sense to define an ideal $J_\U$ which is obtained by stripping away the Boolean relations.  This will then be used in the proof of Theorem~\ref{thm:ideal-equivalence}.

Note that if $\sigma \cap \tau \neq \emptyset$, then for any $i \in \sigma \cap \tau$ we have $U_\sigma \subseteq U_i \subseteq \bigcup_{j \in \tau} U_i$, and the corresponding relation is a multiple of the Boolean relation $x_i(1-x_i)$.
We can thus restrict attention to relations in $I_\U$ that have $\sigma \cap \tau = \emptyset,$ so long as we include separately the Boolean relations.
These observations are summarized by the following lemma.

\begin{lemma}\label{lemma:JU}
$I_\U = J_\U + \langle x_1^2-x_1,\ldots,x_n^2-x_n \rangle,$
where
$$J_\U \od  \big \langle \{ x_\sigma \prod_{i \in \tau} (1-x_i) \mid \sigma \cap \tau = \emptyset \;\;\mathrm{ and }\;\; U_\sigma \subseteq \bigcup_{i \in \tau}U_i \} \big \rangle.$$
\end{lemma}

\begin{proof}[Proof of Theorem~\ref{thm:ideal-equivalence}]
We will show that $J_\U=J_{\C(\U)}$ (and thus that  $I_\U = I_{\C(\U)}$) by showing that each ideal contains the
generators of the other. 

First, we show that all generating relations of $J_{\C(\U)}$ are contained in $J_\U$.  Recall that the generators of $J_{\C(\U)}$ are of the form
$$\rho_v = \prod_{i \in \supp(v)}x_i\prod_{j \notin \supp(v)}(1-x_j) \;\; \text{for} \;\; v\notin \C(\U).$$
If $\rho_v$ is a generator of $J_{\C(\U)}$, then $v \notin \C(\U)$ and this implies (by the definition of $\C(\U)$)
that $U_{\supp(v)} \subseteq \bigcup_{j \notin \supp(v)} U_j$.  Taking $\sigma = \supp(v)$ and $\tau = [n] \setminus \supp(v)$, we have $U_{\sigma} \subseteq \bigcup_{j \in \tau} U_j$ with $\sigma \cap \tau = \emptyset$.  This in turn tells us (by the definition of $J_\U$) that $x_\sigma \prod_{j \in \tau} (1-x_j)$ is a generator of $J_\U$.  Since 
$\rho_v = x_\sigma \prod_{j \in \tau} (1-x_j)$ for our choice of $\sigma$ and $\tau$, we conclude that $\rho_v \in J_\U$.  Hence, $J_{\C(\U)} \subseteq J_\U$.

Next, we show that all generating relations of $J_\U$ are contained in $J_{\C(\U)}$.
If $J_\U$ has generator  $x_\sigma\prod_{i\in \tau}(1-x_i)$, then $U_\sigma\subseteq \bigcup_{i\in \tau} U_i$ and $\sigma \cap \tau = \emptyset$.  This in turn implies that $\bigcap_{i \in \sigma} U_i \setminus \bigcup_{j \in \tau} U_j = \emptyset$, and thus (by the definition of $\C(\U)$) we have $v\notin \C(\U)$ for any $v$ such that $\supp(v)\supseteq \sigma$ and $\supp(v)\cap \tau=\emptyset$.
It follows that $J_{\C(\U)}$ contains the relation $x_{\supp(v)}\prod_{j\notin \supp(v)}(1-x_j)$ for any such $v$.  This includes all relations of the form
  $x_\sigma \prod_{j\in \tau}(1-x_j) \prod_{k\notin \sigma\cup \tau}P_k$, where $P_k\in \{x_k,1-x_k\}$.  Taking 
  $f =  x_\sigma \prod_{j\in \tau}(1-x_j)$ in Lemma~\ref{lemma:induction-arg} (below), we can conclude that $J_{\C(\U)}$ contains $x_\sigma \prod_{j\in \tau}(1-x_j)$.  Hence, $J_\U \subseteq J_{\C(\U)}$.
  \end{proof}

\begin{lemma} \label{lemma:induction-arg}
For any $f\in k[x_1,\ldots,x_n]$ and $\tau\subseteq [n]$, the ideal $\langle
\big\{f\prod_{i\in \tau} P_i \mid P_i \in \{x_i, 1-x_i\}\big\}
\rangle = \langle f\rangle.$
\end{lemma}
\begin{proof}

First, denote
$I_f(\tau)\stackrel{\text{def}}{=}\langle\big\{f\prod_{i\in \tau}
P_i \mid P_i \in \{x_i, 1-x_i\}\big\} \rangle $.  We wish to prove
that $I_f(\tau)=\langle f \rangle$, for any $\tau \subseteq [n]$. Clearly, $I_f(\tau)\subseteq \langle f\rangle$,
since every generator of $I_f(\tau)$ is a multiple of $f$. We will prove
$I_f(\tau)\supseteq \langle f\rangle$ by induction on $|\tau|$.

If $|\tau|=0$, then $\tau=\emptyset$ and $I_f(\tau)=\langle f\rangle$. If $|\tau|=1$, so that $\tau=\{i\}$ for some
$i\in [n]$, then $I_f(\tau)=\langle f(1-x_i), fx_i\rangle$.  Note that
 $f(1-x_i)+fx_i=f$, so $f\in I_f(\tau)$, and thus
 $I_f(\tau)\supseteq \langle f\rangle$.

 Now, assume that for some $\ell \geq 1$ we have $I_f(\sigma)\supseteq \langle f\rangle$ for any $\sigma\subseteq [n]$
 with $|\sigma| \leq \ell$.  If $\ell \geq n$, we are done, so we need only show that if $\ell < n$, 
 then $I_f(\tau)\supseteq \langle f\rangle$ for any $\tau$ of size $\ell + 1$.
 Consider $\tau\subseteq [n]$ with $|\tau|=\ell+1$, and let $j \in \tau$ be
 any element. Define $\tau'=\tau\backslash
 \{j\}$, and note that $|\tau'|=\ell$. By our inductive assumption,
 $I_f(\tau')\supseteq\langle f\rangle$.  We will show that 
 $I_f(\tau)\supseteq I_f(\tau')$, and hence $I_f(\tau)\supseteq
 \langle f\rangle$.

 Let $g = f\prod_{i\in \tau'}P_i$ be any generator of
 $I_f(\tau')$ and observe that both $f(1-x_j)\prod_{i\in \tau'}P_i$ and
 $f x_j\prod_{i\in \tau'}P_i$ are both generators of $I_f(\tau)$.  It follows that 
 their sum, $g$, is also in $I_f(\tau)$, and hence $g \in I_f(\tau)$ 
 for any generator $g$ of  $I_f(\tau')$. We conclude that $I_f(\tau)\supseteq I_f(\tau')$, as desired.
\end{proof}

\subsection{Proof of Theorem~\ref{thm:canonical-form}}\label{sec:proof2}

We begin by showing that $J_\U,$ first defined in Lemma~\ref{lemma:JU}, can be generated using the Type 1, Type 2 and Type 3 relations introduced in Section~\ref{sec:types}.   From the proof of Theorem~\ref{thm:ideal-equivalence},
we know that $J_\U = J_{\C(\U)},$ so the following lemma in fact shows that $J_{\C(\U)}$ is generated by the Type 1, 2 and 3 relations as well.

\begin{lemma} \label{lemma:gen-types}
For $\U = \{U_1,\ldots, U_n\}$ a collection of sets in a stimulus space $X$,
\begin{eqnarray*}
J_{\U}  &=& \big \langle \{x_\sigma \mid U_\sigma = \emptyset \}, \big\{\prod_{i \in \tau} (1-x_i) \mid X \subseteq \bigcup_{i\in \tau} U_i\big\},
\\
&& \big\{x_\sigma \prod_{i\in \tau} (1-x_i) \mid \sigma,\tau \neq \emptyset,\; \sigma \cap \tau =\emptyset,\; U_\sigma \neq \emptyset,\; \bigcup_{i \in \tau} U_i \neq X ,  \text{ and } U_\sigma\subseteq \bigcup_{i\in \tau}U_i  \big\} \big \rangle.
\end{eqnarray*}
$J_\U$ (equivalently, $J_{\C(\U)}$) is thus generated by the Type 1, Type 3 and Type 2 relations, respectively.
\end{lemma}

\begin{proof} Recall that in Lemma~\ref{lemma:JU} we defined $J_\U$ as:
$$J_\U \od  \big \langle \{ x_\sigma \prod_{i \in \tau} (1-x_i) \mid \sigma \cap \tau = \emptyset \;\;\mathrm{ and }\;\; U_\sigma \subseteq \bigcup_{i \in \tau}U_i \} \big \rangle.$$
Observe that if $U_\sigma = \emptyset$, then we can take $\tau = \emptyset$ to
obtain the Type 1 relation $x_\sigma$, where we have used the fact that $\prod_{i \in \emptyset}(1-x_i) = 1$.  Any other relation with $U_\sigma = \emptyset$ and $\tau \neq \emptyset$ would be a multiple of $x_\sigma$.  We can thus write:
$$J_\U = \big \langle \{x_\sigma \mid U_\sigma = \emptyset\}, \{ x_\sigma \prod_{i \in \tau} (1-x_i) \mid  \tau \neq \emptyset, \; \sigma \cap \tau = \emptyset, \;U_\sigma \neq \emptyset, \;\mathrm{ and }\;\; U_\sigma \subseteq \bigcup_{i \in \tau}U_i \} \big \rangle.$$
Next, if $\sigma = \emptyset$ in the second set of relations above, then we have the relation
 $\prod_{i \in \tau} (1-x_i)$ with
$U_\emptyset = X \subseteq \bigcup_{i \in \tau}U_i.$  Splitting off these Type 3 relations, and removing multiples of them that occur if $\bigcup_{i \in \tau} U_i = X $, we obtain the desired result.
\end{proof}

Next, we show that $J_\U$ can be generated by reduced sets of  the Type 1, Type 2 and Type 3 relations given above.
First, consider the Type 1 relations in Lemma~\ref{lemma:gen-types}, and observe that if $\tau\subseteq \sigma$, then $x_\sigma$ is a multiple of $x_\tau$.  We can thus reduce the set of Type 1 generators needed by taking only those corresponding to minimal $\sigma$ with $U_\sigma=\emptyset$:
$$\langle \{x_\sigma \mid U_\sigma = \emptyset\}\rangle = \langle \{x_\sigma \mid \sigma \text{ is minimal w.r.t. } U_\sigma = \emptyset\}\rangle.$$
Similarly, we find for the Type 3 relations:
$$\big\langle \big\{\prod_{i \in \tau} (1-x_i) \mid X \subseteq \bigcup_{i\in \tau} U_i\big\}\big \rangle = 
\big\langle \big\{\prod_{i \in \tau} (1-x_i) \mid \tau  \text{ is minimal w.r.t. } X \subseteq \bigcup_{i\in \tau} U_i\big\}\big \rangle. $$
Finally, we reduce the Type 2 generators.  If $\rho\subseteq \sigma$ and $x_\rho\prod_{i\in \tau} (1-x_i) \in J_\U$, then we also have $x_\sigma\prod_{i\in \tau} (1-x_i) \in J_\U$.  So we can restrict ourselves to only those generators for which $\sigma$ is minimal with respect to $U_\sigma\subseteq \bigcup_{i\in \tau}U_i$.   Similarly, we can reduce to minimal $\tau$ such that $U_\sigma\subseteq \bigcup_{i\in \tau}U_i$.  In summary:
\begin{eqnarray*}
&&\big\langle \big\{x_\sigma \prod_{i\in \tau} (1-x_i) \mid \sigma,\tau \neq \emptyset,\; \sigma \cap \tau =\emptyset,\; U_\sigma \neq \emptyset,\;\bigcup_{i \in \tau} U_i \neq X, \text{ and } U_\sigma\subseteq \bigcup_{i\in \tau}U_i  \big\} \big \rangle = \\
&&\big\langle \big\{x_\sigma \prod_{i\in \tau} (1-x_i) \mid \sigma,\tau \neq \emptyset,\; \sigma \cap \tau =\emptyset,\; U_\sigma \neq \emptyset,\;\bigcup_{i \in \tau} U_i \neq X, \text{ and } \sigma, \tau \text{ are each minimal }\\
&&\text{ w.r.t. } U_\sigma\subseteq \bigcup_{i\in \tau}U_i  \big\} \big \rangle.
\end{eqnarray*}

\noindent We can now prove Theorem~\ref{thm:canonical-form}.

\begin{proof}[Proof of Theorem~\ref{thm:canonical-form}]

Recall that $\C = \C(\U)$, and that by the proof of Theorem~\ref{thm:ideal-equivalence} we have $J_{\C(\U)} = J_\U$. By the reductions given above for the Type 1, 2 and 3 generators, we also know that $J_\U$ can be reduced to the form given in the statement of Theorem~\ref{thm:canonical-form}.  We conclude that $J_\C$ can be expressed in the desired form.

To see that $J_\C$, as given in the statement of Theorem~\ref{thm:canonical-form}, is in canonical form, we must show that the given set of generators is exactly the complete set of minimal
pseudo-monomials for $J_\C$.  First, observe that the generators are all pseudo-monomials.   If $x_\sigma$ is one of the Type 1 relations, and $x_\sigma \in \langle g \rangle$
with $\langle x_\sigma \rangle \neq \langle g \rangle$, then $g = \prod_{i \in \tau} x_i$ for some $\tau \subsetneq \sigma$.  Since $U_\tau \neq \emptyset$, however, it follows that $g \notin J_\C$ and hence $x_\sigma$ is a minimal pseudo-monomial of $J_\C$.  By a similar argument, the Type 2 and Type 3 relations above are also minimal pseudo-monomials in $J_\C$.   

It remains only to show that there are no additional minimal pseudo-monomials in $J_\C$.  Suppose $f = x_\sigma \prod_{i \in \tau}(1-x_i)$ is a minimal pseudo-monomial in $J_\C$.    By Lemma~\ref{lemma:corresp}, $U_\sigma \subseteq \bigcup_{i \in \tau} U_i$ and $\sigma \cap \tau = \emptyset$, so $f$ is a generator in the original definition of $J_\U$ (Lemma~\ref{lemma:JU}).  Since $f$ is a minimal pseudo-monomial of $J_\C$, 
there does not exist a $g \in J_\C$ such that $g =  x_{\sigma'} \prod_{i \in \tau'}(1-x_i)$ with either $\sigma' \subsetneq \sigma$ or $\tau' \subsetneq \tau$.  Therefore, $\sigma$ and $\tau$ are each minimal with respect to $U_\sigma \subseteq \bigcup_{i \in \tau} U_i$.  We conclude that $f$ is one of the generators for $J_\C$ given in the statement of Theorem~\ref{thm:canonical-form}.  It is a minimal Type 1 generator if $\tau = \emptyset$, a minimal Type 3 generator if $\sigma = \emptyset$, and is otherwise a minimal Type 2 generator.  The three sets of minimal generators are disjoint because the Type 1, Type 2 and Type 3 relations are disjoint, provided $X \neq \emptyset$.
\end{proof}

\subsection{Proof of Proposition~\ref{prop:prop1}}\label{sec:proof-prop1}

\noindent Note that every polynomial obtained by the canonical form algorithm is a pseudo-monomial of $J_\C$.  This is because the algorithm constructs products of factors of the form $x_i$ or $1-x_i$, and then reduces them in such a way that no index is repeated in the final product, and there are no powers of any $x_i$ or $1-x_i$ factor; we are thus guaranteed to end up with pseudo-monomials.  Moreover, since the products each have at least one factor in each prime ideal of the primary decomposition of $J_\C$, the pseudo-monomials are all in $J_\C$.   Proposition~\ref{prop:prop1} states that this set of pseudo-monomials is precisely the canonical form $CF(J_\C)$.

To prove Proposition~\ref{prop:prop1}, we will make use of the following technical lemma.
Here $z_i, y_i\in \{x_i, 1-x_i\}$, and thus any pseudo-monomial in $\F_2[x_1,\ldots,x_n]$ is of the form $\prod_{j\in \sigma}  z_j$ for some index set $\sigma \subseteq [n]$.  

\begin{lemma}\label{lemma:matching}
If $y_{i_1}\cdots y_{i_m}\in \langle z_{j_1},\ldots, z_{j_\ell}\rangle$ where $\{i_k\}$ and $\{j_r\}$ are each distinct sets of indices, then $y_{i_k}=z_{j_r}$ for some $k\in [m]$ and $r\in [\ell]$.
\end{lemma}

\begin{proof} 
Let $f = y_{i_1}\cdots y_{i_m}$ and $P = \{ z_{j_1},\ldots,z_{j_\ell} \}$.  
Since $f\in \langle P \rangle$, then $\langle P \rangle  = \langle P , f \rangle$, and so $V(\langle P \rangle ) = V(\langle P , f  \rangle)$.  
We need to show that $y_{i_k} = z_{j_r}$ for some pair of indices $i_k, j_r.$
Suppose by way of contradiction that there is no $i_k, j_r$ such that $y_{i_k}=z_{j_r}$.  

Select $a\in \{0,1\}^n$ as follows: for each $j_r \in \{j_1,\ldots,j_\ell\}$, let $a_{j_r} = 0$ 
if $z_{j_r} = x_{j_r}$, and let $a_{j_r} = 1$ if $z_{j_r} = 1-x_{j_r}$; when evaluating at $a$, we thus have $z_{j_r}(a) = 0$ for all $r \in [\ell]$.
Next, for each $i_k\in \omega \od \{i_1,\ldots,i_m\}\backslash \{j_1,..,j_\ell\}$, let $a_{i_k} = 1$ if $y_{i_k} = x_{i_k}$, and let $a_{i_k} = 0$ if $y_{i_k} = 1-x_{i_k}$, so that
$y_{i_k}(a)=1$ for all $i_k \in \omega$.  For any remaining indices $t$, let $a_t=1$.  
Because we have assumed that $y_{i_k}\neq z_{j_r}$ for any $i_k, j_r$ pair, we have
 for any $i \in \{i_1,\ldots,i_m\} \cap \{j_1,\ldots,j_\ell\}$ that $y_i(a) = 1-z_i(a) = 1.$  It follows that $f(a) = 1$.

Now, note that $a\in V(\langle P\rangle )$ by construction.  We must therefore have $a\in V(\langle P, f \rangle )$, and hence $f(a) = 0$, a contradiction.
We conclude that there must be some $i_k, j_r$ with $y_{i_k}=z_{j_r},$ as desired.
\end{proof}

\noindent  We can now prove the Proposition. 

\begin{proof}[Proof of Proposition~\ref{prop:prop1}]
It suffices to show that after Step 4 of the algorithm, the reduced set $\tilde{\mathcal{M}}(J_\C)$ consists entirely of pseudo-monomials of $J_\C$, and includes all {\it minimal} pseudo-monomials of $J_\C$.  If this is true, then after removing multiples of lower-degree elements in Step 5 we are guaranteed to obtain the set of minimal pseudo-monomials, $CF(J_\C)$, since it is precisely the non-minimal pseudo-monomials that will be removed in the final step of the algorithm.

Let $J_\C = \bigcap_{i=1}^s P_i$ be the primary decomposition of $J_\C$, with each $P_i$ a prime ideal of the form $P_i = \langle z_{j_1},\ldots,z_{j_\ell}\rangle$.
Recall that $\mathcal{M}(J_\C)$, as defined in Step 3 of the algorithm, is precisely the set of all polynomials $g$ that are obtained by choosing one linear factor
from the generating set of each $P_i$:
$$\mathcal{M}(J_\C) = \{g = z_{p_1}\cdots z_{p_s} \mid z_{p_i} \text{ is a linear generator of } P_i \}.$$ 
Furthermore, recall that $\tilde{\mathcal{M}}(J_\C)$ is obtained from $\mathcal{M}(J_\C)$ by the reductions in Step 4 of the algorithm.
Clearly, all elements of $\tilde{\mathcal{M}}(J_\C)$ are pseudo-monomials that are contained in $J_\C$.

To show that $\tilde{\mathcal{M}}(J_\C)$ contains all {\it minimal} pseudo-monomials of $J_\C$, we will show that if $f \in J_\C$ is a pseudo-monomial, then there exists another pseudo-monomial $h \in \tilde{\mathcal{M}}(J_\C)$ (possibly the same as $f$) such that $h | f$.  To see this, let
 $f = y_{i_1}\cdots y_{i_m}$ be a pseudo-monomial of $J_\C$.  Then,
$f \in P_i$ for each $i \in [s].$  For a given $P_i = \langle z_{j_1},\ldots,z_{j_\ell}\rangle,$ by Lemma~\ref{lemma:matching} we have $y_{i_k} = z_{j_r}$ for some $k \in [m]$ and $r \in [\ell]$.  In other words, each prime ideal $P_i$ has a generating term, call it $z_{p_i},$ that appears as one of the linear factors of 
$f$.  Setting $g = z_{p_1}\cdots z_{p_s}$, it is clear that $g \in \mathcal{M}(J_\C)$ and that either $g | f$, or $z_{p_i} = z_{p_j}$ for some distinct pair $i,j$.  By removing repeated factors in $g$ one obtains a pseudo-monomial $h \in \tilde{\mathcal{M}}(J_\C)$ such that $h | g$ and $h | f$.  If we take $f$ to be a minimal pseudo-monomial, we find
$f = h  \in \tilde{\mathcal{M}}(J_\C)$.  
\end{proof}

\subsection{Proof of Lemmas~\ref{new-lemma} and~\ref{lemma:z-decomp}}\label{sec:prop2-proof}

Here we prove Lemmas~\ref{new-lemma} and~\ref{lemma:z-decomp}, which underlie the primary decomposition algorithm.  

\begin{proof}[Proof of Lemma~\ref{new-lemma}]  Assume $f \in \langle J,z \rangle$ is a pseudo-monomial.  Then $f = z_{i_1}z_{i_2}\cdots z_{i_r}$, where $z_i \in \{x_i, 1-x_i\}$ for each $i$, and the $i_k$ are distinct.  Suppose $f \notin \langle z \rangle.$  This implies $z_{i_k} \neq z$ for all factors appearing in $f$.  We will show that either $f \in J$ or $(1-z)f \in J$.

Since $J$ is a pseudo-monomial ideal, we can write 
$$J = \langle z g_1, \ldots, z g_k, (1-z) f_1, \ldots, (1-z) f_l, h_1, \ldots, h_m \rangle,$$
where the $g_j, f_j$ and $h_j$ are pseudo-monomials that contain no $z$ or $1-z$ term.  This means
$$f = z_{i_1}z_{i_2}\cdots z_{i_r} = z \sum_{j=1}^k u_j g_j  + (1-z) \sum_{j=1}^l v_j f_j + \sum_{j=1}^m w_j h_j + y z,$$
for polynomials $u_j, v_j, w_j,$ and $y \in \F_2[x_1,\ldots,x_n]$. 
Now consider what happens if we set $z = 0$ in $f$:
$$f|_{z=0} =  z_{i_1}z_{i_2}\cdots z_{i_r}|_{z=0} = \sum_{j=1}^l v_j|_{z=0} f_j + \sum_{j=1}^m w_j|_{z=0} h_j.$$
Next, observe that after multiplying the above by $(1-z)$ we obtain an element of $J$:
$$(1-z) f|_{z=0} = (1-z) \sum_{j=1}^l v_j|_{z=0} f_j + (1-z) \sum_{j=1}^m w_j|_{z=0} h_j \in J,$$
since $(1-z)f_j \in J$ for $j = 1,\ldots,l$ and $h_j \in J$ for $j=1,\ldots,m$.
There are two cases: 
\begin{itemize}

\item[Case 1:] If $1-z$ is a factor of $f$, say $z_{i_1} = 1-z$, then $f|_{z=0} =  z_{i_2}\cdots z_{i_r}$ and thus
$f = (1-z) f|_{z=0} \in J.$

\item[Case 2:] If $1-z$ is {\it not} a factor of $f$, then $f = f|_{z=0}.$  Multiplying by $1-z$ we obtain 
$(1-z) f  \in J.$

\end{itemize}
We thus conclude that $f \notin \langle z \rangle$ implies $f \in J$ or $(1-z)f \in J$.
\end{proof}
\medskip

\begin{proof}[Proof of Lemma~\ref{lemma:z-decomp}]
Clearly, $\langle J, z_\sigma \rangle \subseteq \bigcap_{i \in \sigma} \langle J, z_i \rangle.$  To see the reverse inclusion, consider
$f \in \bigcap_{i \in \sigma} \langle J, z_i \rangle.$  We have three cases.
\begin{itemize}
\item[Case 1:] $f \in J$.  Then, $f \in \langle J, z_\sigma \rangle.$
\item[Case 2:] $f \notin J$, but $f \in \langle z_i \rangle$ for all $i \in \sigma$.  Then $f \in \langle z_\sigma \rangle$, and hence $f \in \langle J, z_\sigma \rangle.$
\item[Case 3:] $f \notin J$ and $f \notin \langle z_i \rangle$ for all $i \in \tau \subset \sigma$, but $f \in \langle z_j \rangle$ for all $j \in \sigma \setminus \tau$.  
Without loss of generality, we can rearrange indices
so that $\tau = \{1,\ldots,m\}$ for $m \geq 1$.  By Lemma~\ref{new-lemma}, we have $(1-z_i)f \in J$ for all $i \in \tau$.
We can thus write:
$$f = (1-z_1)f + z_1(1-z_2)f + \ldots + z_1\cdots z_{m-1}(1-z_m) f + z_1 \cdots z_m f.$$
Observe that the first $m$ terms are each in $J$.  On the other hand, $f \in  \langle z_j \rangle$ for each $j \in \sigma \setminus \tau$ implies that the last term is in 
$\langle z_\tau \rangle \cap \langle z_{\sigma \setminus \tau} \rangle = \langle z_\sigma \rangle.$  Hence, $f \in \langle J, z_\sigma \rangle.$
\end{itemize}
We may thus conclude that $\bigcap_{i \in \sigma} \langle J, z_i \rangle \subseteq \langle J, z_\sigma \rangle$, as desired.
\end{proof}

\subsection{Proof of Theorem~\ref{thm:prim-decomp}}\label{sec:prim-decomp-proof}

Recall that $J_\C$ is always a proper pseudo-monomial ideal for any nonempty neural code $\C \subseteq \{0,1\}^n$. Theorem~\ref{thm:prim-decomp} is thus a direct consequence of the following proposition.

\begin{proposition}\label{prop:primarydec}
Suppose $J \subset \F_2[x_1,\ldots,x_n]$ is a proper pseudo-monomial ideal. Then, $J$ has a unique irredundant primary decomposition of the form
$J = \bigcap_{a \in \A} \p_a,$
where $\{\p_a\}_{a \in \A}$ are the minimal primes over $J$.
\end{proposition}

\begin{proof}
By Proposition~\ref{prop:prim-decomp}, we can always (algorithmically) obtain an irredundant set $\P$ of prime ideals such that $J = \bigcap_{I \in \P} I$.
Furthermore, each $I \in \P$ has the form $I = \langle z_{i_1},\ldots,z_{i_k}\rangle$, where $z_i \in \{ x_i, 1-x_i\}$ for each $i$.
Clearly, these ideals are all prime ideals of the form $\p_a$ for $a \in \{0,1,*\}$.  
It remains only to show that this primary decomposition is unique, and that
the ideals $\{\p_a\}_{a \in \A}$ are the minimal primes over $J$.  
This is a consequence of some well-known facts summarized in Lemmas~\ref{lemma:inter} and~\ref{lemma:rad}, below.  First, observe by Lemma~\ref{lemma:inter} that $J$ is a radical ideal.  Lemma~\ref{lemma:rad} then tells us that the decomposition in terms of minimal primes is the unique irredundant primary decomposition for $J$.
\end{proof}

\begin{lemma}\label{lemma:inter}
If $J$ is the intersection of prime ideals, $J=\bigcap_{i=1}^\ell \mathbf{p}_i$, 
then $J$ is a radical ideal. 
\end{lemma}
\begin{proof}
Suppose $p^n\in J$. Then $p^n\in \mathbf{p}_i$ for all $i \in [\ell]$, and hence $p\in \mathbf{p}_i$ for all $i  \in [\ell]$. Therefore, $p\in J$.
\end{proof}

The following fact about the primary decomposition of radical ideals is true over any field, as a consequence of the Lasker-Noether theorems 
\cite[pp. 204-209]{cox-little-oshea}.

\begin{lemma}\label{lemma:rad}
If $J$ is a proper radical ideal, then it has a unique irredundant primary decomposition consisting of the minimal prime ideals over $J$.
\end{lemma}

\newpage

\section{Appendix 2: Neural codes on three neurons}\label{sec:appendix2}

\vspace{-.2in}

\begin{table}[!h]
\begin{small}
\begin{tabular}{l | l | l}
Label & Code  $\C$ & Canonical Form $CF(J_\C)$ \\
 \hline
A1 & 000,100,010,001,110,101,011,111 & $\emptyset$\\
A2 & 000,100,010,110,101,111 & $x_3(1-x_1)$\\
A3 & 000,100,010,001,110,101,111 & $x_2x_3(1-x_1)$\\
A4 & 000,100,010,110,101,011,111 & $x_3(1-x_1)(1-x_2)$\\
A5 & 000,100,010,110,111 & $x_3(1-x_1), \,x_3(1-x_2)$\\
A6 & 000,100,110,101,111 & $x_2(1-x_1), \,x_3(1-x_1)$\\
A7 & 000,100,010,101,111 & $x_3(1-x_1),\, x_1x_2(1-x_3)$\\
A8 & 000,100,010,001,110,111 & $x_1x_3(1-x_2), \, x_2x_3(1-x_1)$\\
A9 & 000,100,001,110,011,111 & $x_3(1-x_2), \, x_2(1-x_1)(1-x_3)$\\
A10 & 000,100,010,101,011,111& $x_3(1-x_1)(1-x_2), \,x_1x_2(1-x_3)$\\
A11 & 000,100,110,101,011,111 & $x_2(1-x_1)(1-x_3), \,x_3(1-x_1)(1-x_2)$\\
A12 & 000,100,110,111 & $x_3(1-x_1),\, x_3(1-x_2), \,x_2(1-x_1)$\\
A13 & 000,100,010,111 & $x_3(1-x_1), \,x_3(1-x_2), \,x_1x_2(1-x_3)$\\
A14 & 000,100,010,001,111 & $x_1x_2(1-x_3), \, x_2x_3(1-x_1), \, x_1x_3(1-x_2)$\\
A15 & 000,110,101,011,111  & $x_1(1-x_2)(1-x_3), \, x_2(1-x_1)(1-x_3), \, x_3(1-x_1)(1-x_2)$ \\
A16* & 000,100,011,111 & $x_2(1-x_3),\, x_3(1-x_2)$\\
A17* & 000,110,101,111 & $x_2(1-x_1), \,x_3(1-x_1),\, x_1(1-x_2)(1-x_3)$\\
A18* & 000,100,111 & $x_2(1-x_1), \, x_2(1-x_3), \, x_3(1-x_1), \, x_3(1-x_2)$\\
A19* & 000,110,111 & $x_3(1-x_1), \, x_3(1-x_2), \, x_1(1-x_2), \, x_2(1-x_1)$\\
A20* & 000,111 & $x_1(1-x_2), \, x_2(1-x_3), \,x_3(1-x_1), \, x_1(1-x_3), \,x_2(1-x_1), \,x_3(1-x_2)$\\
\hline
B1 & 000,100,010,001,110,101 & $x_2 x_3$ \\
B2 & 000,100,010,110,101& $x_2x_3, \,x_3(1-x_1)$\\
B3 & 000,100,010,101,011 & $x_1x_2, \, x_3(1-x_1)(1-x_2)$\\
B4 & 000,100,110,101 & $x_2x_3, \, x_2(1-x_1), \, x_3(1-x_1)$\\
B5 & 000,100,110,011 & $x_1x_3, \, x_3(1-x_2), \, x_2(1-x_1)(1-x_3)$ \\
B6* & 000,110,101 & $x_2x_3, \, x_2(1-x_1), \, x_3(1-x_1), \,x_1(1-x_2)(1-x_3)$\\
\hline
C1 & 000,100, 010,001, 110 & $x_1x_3, \, x_2x_3$\\
C2 & 000,100,010,101 & $x_1x_2, \, x_2x_3, \, x_3(1-x_1)$\\
C3* & 000,100,011 & $x_1x_2, \, x_1x_3, \, x_2(1-x_3), \, x_3(1-x_2)$\\
\hline
D1 & 000,100,010,001 & $x_1x_2, \, x_2x_3, \, x_1x_3$\\
\hline
E1& 000,100,010,001,110,101,011 & $x_1x_2x_3 $\\
E2 & 000,100,010,110,101,011 & $x_1x_2x_3, \, x_3(1-x_1)(1-x_2)$ \\
E3 & 000,100,110,101,011 & $x_1x_2x_3, \, x_2(1-x_1)(1-x_2), \, x_3(1-x_1)(1-x_2)$\\
E4 & 000,110,011,101 & $x_1x_2x_3, \, x_1(1-x_2)(1-x_3), \, x_2(1-x_1)(1-x_3), \, x_3(1-x_1)(1-x_2)$\\
\hline
F1* & 000,100,010,110 & $x_3$ \\
F2* & 000,100,110 & $x_3, \, x_2(1-x_1)$\\
F3* & 000,110 & $x_3, \,x_1(1-x_2), \, x_2(1-x_1)$\\
\hline 
G1* & 000,100 & $x_2, \, x_3$\\
\hline
H1* & 000 & $x_1, \, x_2, \,  x_3$\\
\hline
I1* & 000,100,010 & $x_3, \, x_1x_2$\\
\end{tabular}
\end{small}
\caption{\small
Forty permutation-inequivalent codes, each containing $000$, on three neurons.  Labels A--I indicate the various families of Type 1 relations present in $CF(J_\C)$, organized as follows (up to permuation of indices):
(A) None,
(B) $\{x_1x_2\}$,
(C) $\{x_1x_2, x_2x_3\}$,
(D) $\{x_1x_2, x_2x_3, x_1x_3\}$,
(E) $\{x_1x_2x_3\}$,
(F) $\{x_1\}$,
(G) $\{x_1, x_2\}$,
(H) $\{x_1, x_2, x_3\}$,
(I) $\{x_1, x_2x_3\}$.
All codes within the same A--I series share the same simplicial complex, $\Delta(\C)$.
The $*$s denote codes that have $U_i = \emptyset$ for at least one receptive field (as in the F, G, H and I series) as well as codes that require $U_1 = U_2$ or $U_1 = U_2 \cup U_3$ (up to permutation of indices); these are considered to be highly degenerate.  The remaining 27 codes are depicted with receptive field diagrams (Figure 6) and Boolean lattice diagrams (Figure 7).
}
\end{table}

\newpage

\begin{figure}[!h]
\begin{center}
\includegraphics[width=6in]{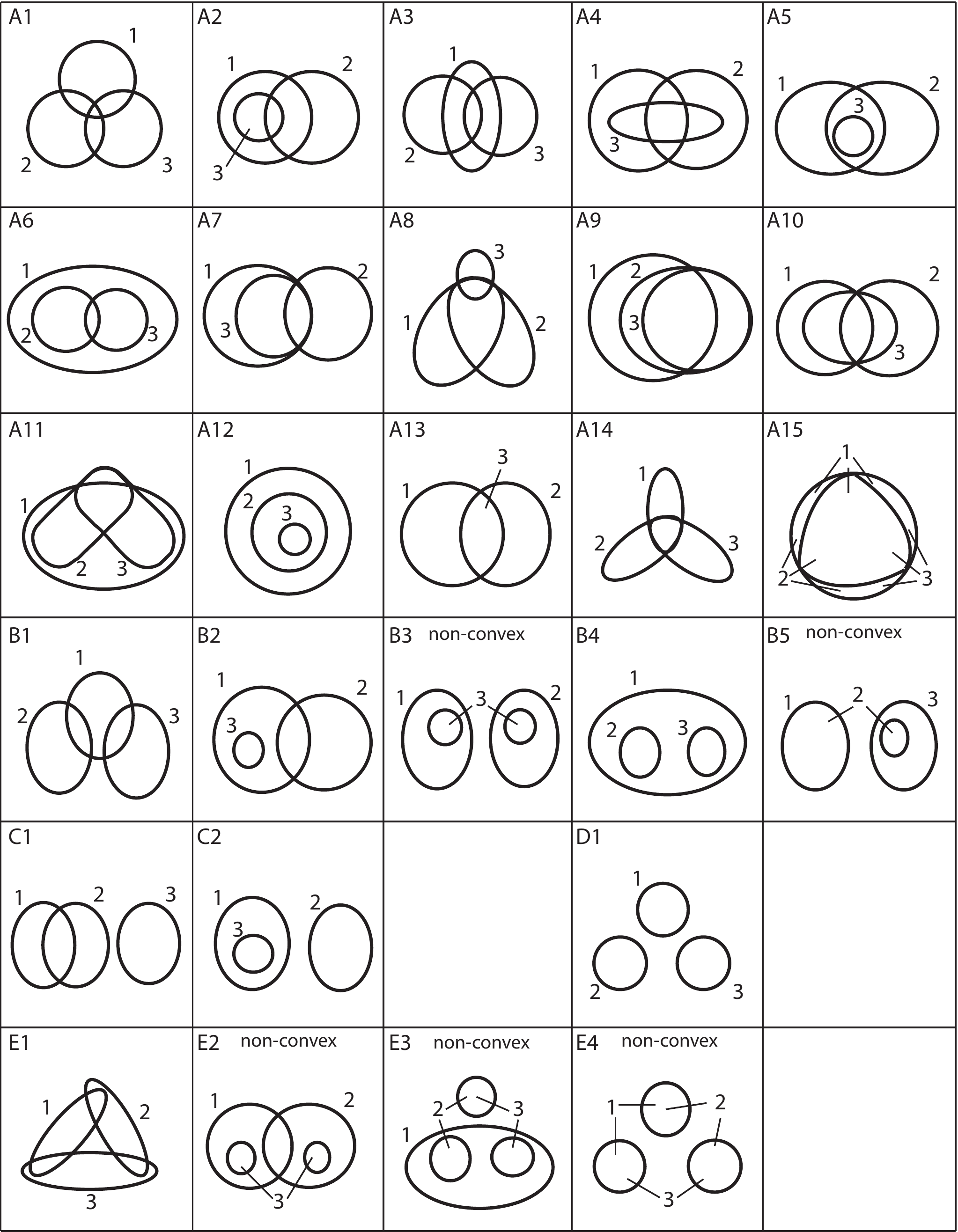}
\end{center}
\vspace{-.2in}
\caption{\small Receptive field diagrams for the 27 non-$*$ codes on three neurons listed in Table 1.  
Codes that admit no realization as a convex RF code are labeled ``non-convex.''  The code E2 is the one from Lemma~\ref{lemma:convex-counterexample}, while
A1 and A12 are permutation-equivalent to the codes in Figure 3A and 3C, respectively.  Deleting the all-zeros codeword from A6 and A4 yields codes permutation-equivalent 
to those in Figure 3B and 3D, respectively.}
\end{figure}

\newpage

\begin{figure}[!h]
\begin{center}
\includegraphics[width=6in]{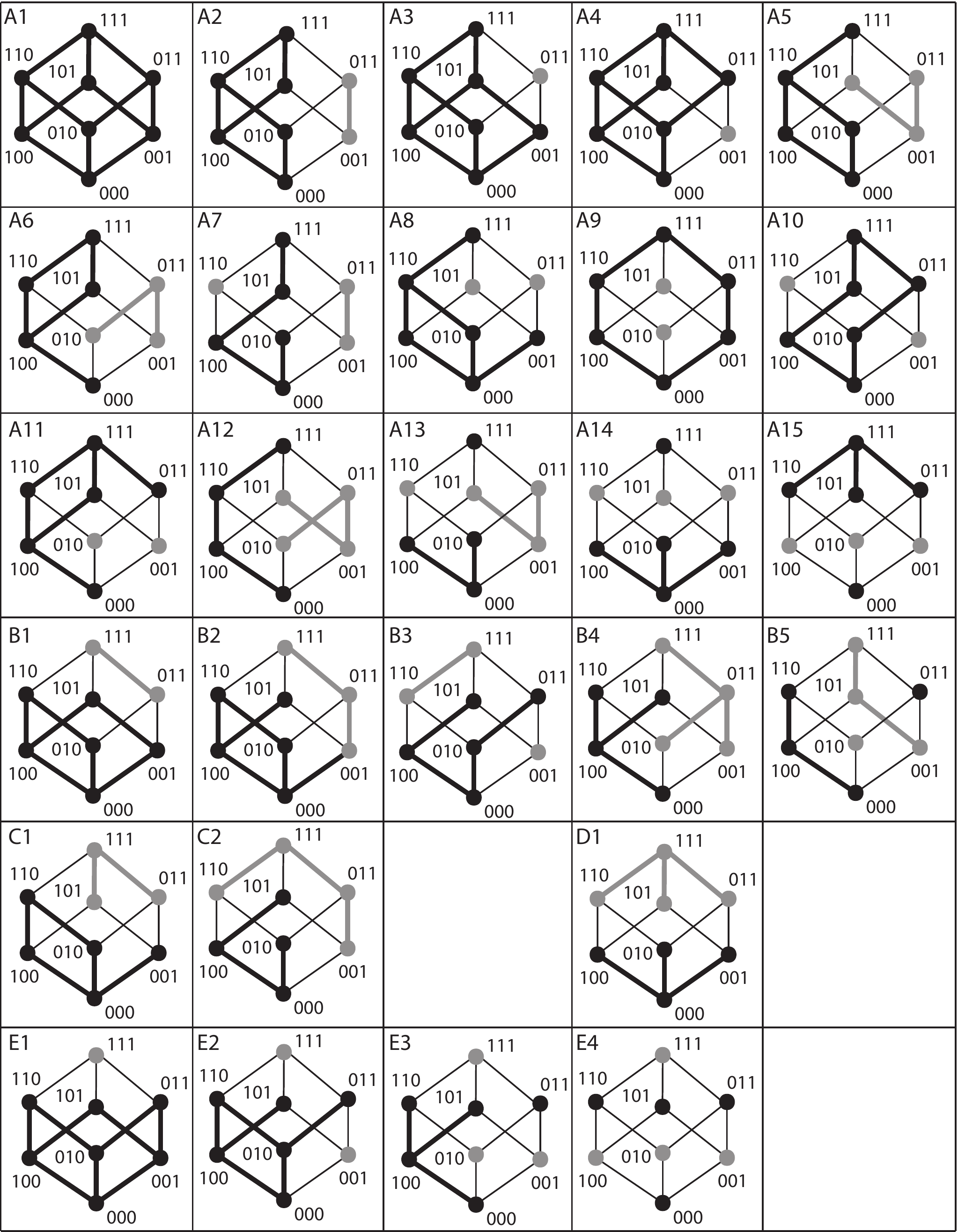}
\end{center}
\vspace{-.2in}
\caption{\small Boolean lattice diagrams for the 27 non-$*$ codes on three neurons listed in Table 1.  Interval decompositions (see Section~\ref{sec:boolean-lattice}) 
for each code are depicted in black,
while decompositions of code complements, arising from $CF(J_\C)$, are shown in gray.  Thin black lines connect elements of the Boolean lattice
that are Hamming distance 1 apart.  Note that the lattice in A12 is permutation-equivalent to the one depicted in Figure 5.}
\end{figure}

\newpage

\bibliographystyle{unsrt}
\bibliography{neural-ring-references}

\end{document}